\def\comments{0}
\def\submission{0}

\ifnum\submission=0
\documentclass[11pt,letterpaper]{article}
\usepackage[margin=1in]{geometry}
\else
\documentclass{llncs}
\fi
\usepackage[utf8]{inputenc}
\usepackage{xcolor}
\usepackage{amsmath}
\usepackage{amsfonts}
\ifnum\submission=0
\usepackage{amsthm}
\fi
\usepackage{braket}
\usepackage{algorithm}
\floatname{algorithm}{Procedure}
\usepackage{algorithmic}
\usepackage{subcaption}
\usepackage{authblk}
\usepackage{comment}
\usepackage[colorlinks=true,allcolors=blue]{hyperref}
\usepackage{breakcites}
\usepackage{cleveref}
\usepackage[sanserif,basic]{complexity}
\usepackage{graphicx}

\ifnum\submission=0
\newtheorem{theorem}{Theorem}[section]

\newtheorem{lemma}[theorem]{Lemma}
\newtheorem{claim}[theorem]{Claim}

\newtheorem{corollary}[theorem]{Corollary}
\newtheorem{definition}[theorem]{Definition}

\newtheorem{remark}{Remark}

\fi
\DeclareMathOperator*{\Ex}{\mathbb{E}}

\crefname{claim}{Claim}{Claims}
\ifnum\comments=1
    \newcommand{\takashi}[1]{{\color{red}(\textbf{Takashi}: #1)}}
    \newcommand{\xingjian}[1]{{\color{magenta}(\textbf{Xingjian}: #1)}}
    \newcommand{\angelos}[1]{{\color{blue}(\textbf{Angelos}: #1)}}
    \newcommand{\qipeng}[1]{{\color{teal}(\textbf{Qipeng}: #1)}}
\else
    \newcommand{\takashi}[1]{}
    \newcommand{\xingjian}[1]{}
    \newcommand{\angelos}[1]{}
    \newcommand{\qipeng}[1]{}
\fi

\newcommand{\ra}{\rightarrow}

\newcommand{\game}{\mathcal{G}}
\newcommand{\flag}{\mathsf{flag}}
\newcommand{\good}{\mathsf{Good}}
\newcommand{\la}{\leftarrow}
\newcommand{\negl}{\mathsf{negl}}
\newcommand{\secp}{n}

\newcommand{\A}{\mathcal{A}}
\newcommand{\B}{\mathcal{B}}
\newcommand{\M}{\mathcal{M}}
\newcommand{\V}{\mathcal{V}}
\newcommand{\bit}{\{0,1\}}
\newcommand{\ch}{\mathsf{CH}}
\newcommand{\Func}{\mathsf{Func}}
\newcommand{\dist}{\mathcal{D}}
\newcommand{\Ver}{\mathsf{Ver}}

\newcommand{\vv}{\mathbf{v}}
\newcommand{\calL}{\mathcal{L}}
\newcommand{\ora}{\mathcal{O}}

\newcommand{\CPhO}{\mathsf{CPhO}}
\newcommand{\CStO}{\mathsf{CStO}}
\newcommand{\F}{\mathcal{F}}
\newcommand{\samp}{{\sf Samp}}
\newcommand{\query}{{\sf Query}}
\newcommand{\ver}{{\sf Verify}}
\newcommand{\calH}{\mathcal{H}}
\newcommand{\calD}{\mathcal{D}}

\newcommand{\As}{\mathcal{A}}
\newcommand{\Bs}{\mathcal{B}}
\newcommand{\Cs}{\mathcal{C}}
\newcommand{\ans}{{\sf ans}}

\title{Classical vs Quantum Advice and Proofs under Classically-Accessible Oracle}
\ifnum\submission=0
\author[1]{Xingjian Li}
\author[2]{Qipeng Liu}
\author[3]{Angelos Pelecanos}
\author[4]{Takashi Yamakawa}
\affil[1]{Tsinghua University}
\affil[2]{Simons Institute}
\affil[3]{UC Berkeley}
\affil[4]{NTT Social Informatics Laboratories}
\else
\fi

\begin{document}
\maketitle
\ifnum\submission=1
\vspace{-7mm}
\fi

\begin{abstract}
It is a long-standing open question to construct a classical oracle relative to which $\BQP/\qpoly\neq \BQP/\poly$ or $\QMA\neq\QCMA$. 
In this paper, we construct \emph{classically-accessible classical oracles} relative to which $\BQP/\qpoly\neq \BQP/\poly$  and $\QMA\neq\QCMA$. 
Here,  classically-accessible classical oracles are oracles that can be accessed only classically even for quantum algorithms. Based on a similar technique, we also show an alternative proof for the separation of $\QMA$ and $\QCMA$ relative to a distributional quantumly-accessible classical oracle, which was recently shown by Natarajan and Nirkhe.  
\end{abstract}


\section{Introduction}


Quantum information is often in possession of richer structures than classical information, at least intuitively. The first (but often false) thought is that phases and magnitudes are continuous, and a piece of quantum information may be able to store exponentially or infinitely more information than classical ones; which is always not true\footnote{As storing and extracting information takes resources that scale with accuracy.}. Since classical and quantum information present distinct and unique natures, the community studies their differences under different contexts and directions, including advice-aided quantum computation~\cite{NY04,Aaronson05,Aaronson07,SICOMP:AarDru14,nayebi2014quantum,hhan2019quantum,chung2019lower,chung2020tight,TCC:GLLZ21,EC:Liu23}, 
{\sf QMA} vs {\sf QCMA} (i.e., quantum {\sf NP} with either quantum or classical witness)~\cite{AN02,aaronson2007quantum,FK18,NN22}, 
quantum vs classical 
communication complexity~\cite{FOCS:Yao93,STOC:BuhCleWig98,STOC:Raz99,SICOMP:ASTVW03,STOC:BarJayKer04,gavinsky2008classical} 
and many others. 

One way to understand their differences is by studying one-way communication complexity: i.e., Alice and Bob want to jointly compute a function with their private inputs, but only one-time quantum/classical communication from Alice to Bob is allowed. 
Among many works, Bar-Yossef, Jayram, and Kerenidis~\cite{STOC:BarJayKer04} showed an exponential separation between quantum and classical one-way communication complexity, for the so-called hidden matching problem. 

The other approach is by looking at {\sf QMA} vs {\sf QCMA}. In 2007, Aaronson and Kuperberg \cite{aaronson2007quantum} showed a black-box separation with respect to a black-box quantum unitary and left the same separation with respect to a classical oracle as an open question. More than a decade later, Fefferman and Kimmel \cite{FK18} proved a second black-box separation using a distributional in-place oracle, which is a non-standard type of oracles. 
Recently, Natarajan and Nirkhe \cite{NN22} moved a step closer to the goal by presenting a black-box separation with respect to a distributional oracle\footnote{In this case, the witness only depends on the distribution. The oracle is later picked from the distribution, but independent of the witness.}. 
Therefore, we would like to further investigate the difference between quantum and classical proofs, i.e., the separation between {\sf QMA} vs {\sf QCMA}. In this work, we address the question by demonstrating a separation relative to a \emph{classically accessible classical} oracle.

\begin{definition}[{\sf QMA}]
\qipeng{new}
A language $\cal L$ is said to be in {\sf QMA} if there exists a quantum polynomial-time machine $\V$ together with a polynomial $p(\cdot)$ such that,
\begin{itemize}
    \item For all $x \in {\cal L}$, there exists a quantum state $\ket{\psi_x}$ of at most $p(|x|)$ qubits, such that $\V$ accepts on $\ket x, \ket {\psi_x}$ with a probability at least $2/3$. 
    \item For all $x \not\in {\cal L}$, for all quantum states $\ket{\psi_x}$ of at most $p(|x|)$ qubits, $\V$ accepts on $\ket x, \ket {\psi_x}$ with a probability at most $1/3$. 
\end{itemize}
\end{definition}
\noindent One can similarly define the class {\sf QCMA} except $\ket{\psi_x}$ are of $p(|x|)$ classical bits. 

\qipeng{previously, I talked about the intro with the following stories: there are studies on understanding quantum info vs classical info, including one-way communication complexity, quantum proofs and advice-aided computation. I briefly talked about the first two, and the work will focus on the third one. 

Now, since we are also working on quantum proofs, I add the results when I was talking about quantum vs classical proofs. So, I expand this paragraph. }

\medskip

We also aim to understand the difference between quantum and classical information through advice-aided quantum computation. Classically, a piece of advice can significantly speed up classical computation, from speeding up exhaustive search \cite{hellman1980cryptanalytic} to deciding the unary Halting Problem. In a quantum world, advice can be either a piece of classical or quantum information. It is very natural to ask the question: does quantum advice ``outperform’’ classical advice? Among many questions, one of the most fundamental is to understand the power of {\sf BQP/qpoly} and {\sf BQP/poly}: i.e., the class of languages that can be decided by bounded-error quantum machines with arbitrary bounded-length quantum/classical advice and polynomial time. 

\begin{definition}[{\sf BQP/qpoly}]
A language $\cal L$ is said to be in {\sf BQP/qpoly} if and only if there exists 
a quantum polynomial-time machine $\A$ together with a collection of polynomial-sized quantum states $\{\ket{z_n}\}_{n\in \mathbb{N}}$ such that, 
\begin{itemize}
    \item For all $x \in {\cal L}$, $\Pr\left[ \A(x, \ket{z_{|x|}}) = 1 \right] \geq 2/3$.
    \item For all $x \not\in {\cal L}$, $\Pr\left[ \A(x, \ket{z_{|x|}}) = 0 \right] \geq 2/3$.
\end{itemize}
\end{definition}
\noindent One can similarly define the class {\sf BQP/poly} except 
$\ket{z_n}$
are poly-sized classical strings.

Similar to the case of {\sf QMA} vs {\sf QCMA}, Aaronson and Kuperberg~\cite{aaronson2007quantum} in the same paper showed an oracle separation between these two classes, leaving the separation with respect to a classical oracle as an open question. Recently, Liu~\cite{EC:Liu23} showed the separation for its relational variants (i.e., {\sf FBQP/qpoly} and {\sf FBQP/poly}) under a special case, where the oracle is never given to the algorithms\footnote{As mentioned in \Cref{sec:concurrent}, a concurrent work by Aaronson, Buhrman, and Kretschmer \cite{ABK23} proves their relational variants $\mathsf{FBQP}/\qpoly\neq \mathsf{FBQP}/\poly$ \emph{unconditionally}.}. Despite all the efforts, the separation between ${\sf BQP/poly}$ and ${\sf BQP/qpoly}$ relative to a classical oracle remains obscure. 

In this work, we proceed with the question by showing a \emph{full separation} relative to a \emph{classically accessible classical} oracle. 
Along the way, we adapt our techniques and give an alternative proof for the separation between {\sf QMA} vs {\sf QCMA} relative to a \emph{quantumly accessible classical} distributional oracle, which was recently  established by Natarajan and Nirkhe~\cite{NN22}. 

\subsection{Our Results}

Our first result is a black-box separation between {\sf BQP/qpoly} vs {\sf BQP/poly} with respect to a classically-accessible classical oracle. 

\emph{Classically-accessible classical oracles.} A classical oracle $\cal O$ is said to be classically accessible if a quantum algorithm can only query the oracle classically; in other words, the only interface of ${\cal O}$ to quantum algorithms is classical: given an input $x$, it outputs $y = {\cal O}(x)$. 

\begin{theorem}[Informal]\label{thm:main_bqp}
There exists a language $\mathcal{L}$ in {\sf BQP/qpoly} but not in {\sf BQP/poly}, relative to a classically accessible classical oracle $\cal O$.
\end{theorem}


Our work is based on the previous works on quantum advantages with unstructured oracles, by Yamakawa and Zhandry~\cite{FOCS:YamZha22} and recent separation by Liu~\cite{EC:Liu23}. Our work improves \cite{EC:Liu23} in two aspects: first, Liu only proved the separation of their relational variants, instead of the original decision classes; second, the separation by Liu does not allow algorithms to have access to the oracle (either classically, or quantumly), but only the advice can depend on the oracle. 

\begin{remark}
Although the long-standing open question is to understand the separation relative to a quantumly accessible classical oracle, our theorem is not weaker but incomparable. Since classical access is not stronger than quantum access, it limits both the computational power of quantum machines with quantum and classical advice: intuitively, classical access makes it easier to prove a language ${\cal L}$ is not in ${\sf BQP^{\cal O}/poly}$ but harder to prove ${\cal L} \in {\sf BQP^{\cal O}/qpoly}$. A similar observation is applicable to the theorem concerning the separation between {\sf QMA} and {\sf QCMA}.
\end{remark}

Our second result is about the separation between {\sf QMA} and {\sf QCMA}. 
\begin{theorem}[Informal]\label{thm:main_qma_qcma}
There exists a language $\mathcal{L}$ in {\sf QMA} but not in {\sf QCMA}, relative to a classically accessible classical oracle $\cal O$.
\end{theorem}

Inspired by the techniques used in our \Cref{thm:main_bqp} and \Cref{thm:main_qma_qcma}, we give an alternative proof for the result by Natarajan and Nirkhe \cite{NN22}. Note in the following result, the classical oracle is \emph{quantumly accessible}. 
\begin{theorem}[Informal]\label{thm:main_qma}
There exists a language $\mathcal{L}$ in {\sf QMA} but not in {\sf QCMA}, relative to a distributional quantumly-accessible classical oracle $\cal O$. 
\end{theorem}

A difference between the construction of~\cite{NN22} and ours is in the case when the oracle is fixed to some oracle from the distribution. In~\cite{NN22}, they can neither prove nor disprove the separation between {\sf QMA} and {\sf QCMA}. However, in our construction, the separation will disappear if we fix our oracle.

Finally, we observe that the problem considered in \cite{FOCS:YamZha22} gives a new superpolynomial separation between classical and quantum one-way communication complexity. Though such a separation has been known since 2004~\cite{STOC:BarJayKer04}, the new separation has two interesting features that Bob's input length is short and the classical lower bound holds even if Bob can classically access Alice's input as an oracle (see \Cref{sec:communication} for more detail).\footnote{As mentioned in \Cref{sec:concurrent}, a concurrent work by Aaronson, Buhrman, and Kretschmer \cite{ABK23} observes that there is a variant of \cite{STOC:BarJayKer04} that satisfies the former (but not latter).} 

\subsection{Overview}

\paragraph{Quantum Advantages and Separation for a Special Case.}
We will introduce the recent work on quantum advantages from unstructured oracles by Yamakawa and Zhandry~\cite{FOCS:YamZha22}, as it will be used for both of our results in this work. They proved that there exists a code $C \subseteq \Sigma^n$ and an oracle-aided function $f_C$ such that, relative to a random oracle from $H: [n] \times \Sigma \to \{0,1\}$, the function $f_C^H$ is easy to invert on any image with quantum access but inversion is hard with only classical access. The function is defined as the following: 
\begin{align*}
    f^H_C(v_1, \cdots, v_n) = \begin{cases}
                H(1, v_1) || H(2, v_2) || \cdots || H(n, v_n)  & \text{ if } (v_1, \cdots, v_n) \in C \\
                \bot  & \text{ if } (v_1, \cdots, v_n) \not\in C \\
            \end{cases}.
\end{align*}
Intuitively, although the function computes an entry-by-entry hash, the requirement that $(v_1,\cdots, v_n)$ must be a codeword enforces the hardness of inversion when only classical queries are allowed. More precisely, the underlying code $C$ satisfies a property called list-recoverability; even if a classical algorithm learns hash values of a subset $E_i \subseteq \Sigma$ for each of $H(i,\cdot)$, only a polynomial number of codewords can be found in $E_1 \times E_2 \times \cdots \times E_n$, which does not help invert random images\footnote{When all classical queries are non-adaptive, this is clearly true: as only polynomially many $f^H_c$ are known for codewords in $C$. The idea can be adapted to adaptive queries as well; we do not elaborate on it here. }. On the other hand, they showed a quantum algorithm that uses quantum queries to invert images.

Liu~\cite{EC:Liu23} observed that the inversion quantum algorithm only needs to make non-adaptive quantum queries that are independent of the image $y$. Therefore, the algorithm with quantum access can be easily cast into a quantum algorithm with quantum advice but no queries; on the other hand, he showed that, if an algorithm has no access to the random oracle, it can not invert even with a piece of exponentially large classical advice. Since given an image $y$ there are multiple pre-images of $y$, the above two statements lead to the separation between {\sf FBQP/qpoly} and {\sf FBQP/poly} when an algorithm has no access to the oracle. 

\paragraph{Allowing classical queries.}
Our first result is to extend the previous separation of {\sf FBQP/qpoly} and {\sf FBQP/poly} by Liu, by allowing quantum algorithms to make online \emph{classical queries} to $H$. Since the algorithm with quantum advice makes no queries, it also works in this setting of classical access. We only need to prove the hardness with classical advice: i.e., no quantum algorithms can invert with classical queries and bounded classical advice. 

Following the framework by Guo et al.~\cite{TCC:GLLZ21}, when only making classical queries (say, at most $T$), a piece of $S$ bits of classical advice is equivalent to the so-called ``ordinary advice'': the advice consists of only $S T$ coordinates, or $S T$ pairs of inputs and outputs. More precisely, the information the algorithm can learn from $S$ bits of advice and $T$ classical queries is roughly the same as that from $ST$ bits of ordinary advice and $T$ classical queries. Thus, the first step is to replace classical advice with ordinary classical advice. 

It now remains to show that a quantum algorithm with classical access to $H$ and short ordinary advice can not invert random images. As ordinary advice only gives information on at most polynomially many pairs of inputs and outputs, let $E_i$ denote the subset of inputs whose hash values under $H(i, \cdot)$ are in the ordinary advice; since the ordinary advice is of length polynomial, $|E_i|$ is polynomial for each $i \in [n]$. Now let us assume the algorithm makes non-adaptive queries that are independent of a challenge. In this case, we can further define $E'_i$ that consists of all $x$ inputs whose values $H(i,x)$ are known from these classical queries. The algorithm in total learns hashes of $H(i,\cdot)$ for the inputs in $E_i \cup E'_i$. Observing that $|E_i\cup E'_i|$ is polynomial for each $i$, by the list-recoverability of the underlying code $C$, the algorithm only learns values of $f^H_C$ for a polynomial number of codewords, which almost always never hits a random challenge. 

Lastly, we extend the proof to adaptive cases (for more details, please refer to \Cref{sec:YZ}).

\paragraph{Upgrading to {\sf BQP/qpoly} v.s. {\sf 
BQP/poly}.}
Next, we turn the above separation of relational classes into a separation of {\sf BQP/qpoly} v.s. {\sf 
BQP/poly}. Our idea is to define a language ${\cal L}$ through a family of random functions for each $n$:  $G_n\colon\bit^n\to\bit$, i.e., for $x\in\bit^n$, $x \in {\cal L}$ if and only if $G_n(x) = 1$. We omit the subscript when it is clear from context in the introduction.

We use an oracle $\cal O$ to hide the evaluation of $G$ on $x$, by requiring the algorithm to invert the function $f^H_C$ on image $x$.  More precisely, we define $\ora$ as follows: 
\begin{align*}
    \ora(\vv, x) = \begin{cases}
                G(x)  & \text{ if } f^H_C(\vv) = x, \\
                \bot  & \text{ otherwise}.
    \end{cases}
\end{align*}
Then an algorithm only gets oracle access to $\ora$, but not $H$ or $G$.

To see that $\calL$ is decidable by a $\BQP/\qpoly$ machine, we can just take the quantum advice as in~\cite{EC:Liu23}. On an input $x$, an algorithm generates $\vv$ such that $f^H_C(\vv)=x$ using the quantum advice. We can then evaluate $G(x)$ by querying $\ora$ at $(x,\vv)$ and decide if $x \in \cal L$. 

On proving that $\calL$ cannot be decided by any $\BQP/\poly$ machine, we leverage the statement proved above: given classical advice, by only querying a classical oracle $H$, it is hard for an efficient algorithm to invert $f^H_C$. The beyond result implies that the algorithm should only have negligible query weight on $\vv$ under the classical oracle $\ora$ such that $f^H_C(\vv)=x$; otherwise the algorithm can be turned into another one that inverts $f^H_C$. 

Therefore, we can reduce the problem to the case where an algorithm has only access to $G(y)$ for all $y \ne x$, but no access to $G(x)$ for the challenge $x$; the goal is still to learn $G(x)$. This is exactly the famous Yao's box problem~\cite{STOC:Yao90}: a piece of advice is allowed to depend on the whole oracle $G$, but then an online algorithm uses the advice to find $G(x)$ for a random $x$, with no access to $G(x)$. By adapting the ideas in~\cite{STOC:Yao90,C:DeTreTul10} and combining all the previous ideas with a standard diagonalization argument, we prove the separation ${\sf BQP^\ora/qpoly} \ne {\sf BQP^\ora/poly}$ relative to a classically accessible classical oracle.



\paragraph{Separation between {\sf QMA} and {\sf QCMA} relative to a classically accessible classical oracle.}

We first construct a problem that has a short quantum proof for {\sf YES} instances, no quantum proof for {\sf NO} instances, and no classical proof that can distinguish between {\sf YES} and {\sf NO} instances. Given random functions $H: [n] \times \Sigma \to \{0,1\}$ and $G : \{0,1\}^n \to \{0,1\}^n$, as well as a subset $S \subseteq \{0,1\}^n$ with a size of at least $2/3\cdot 2^n$, we create two pairs of oracles:

\begin{itemize}
\item For a {\sf YES} instance, the oracle $G$ and an oracle $\ora[G, H, \emptyset]$ are provided. The latter takes an input $t \in \{0,1\}^n$ and a vector $\vv \in C$ and outputs $1$ if and only if $f^H(\vv) = G(t)$.

\item If it is a {\sf NO} instance, the oracle $G$ and an oracle $\ora[G, H, S]$ is given. The latter takes as input a $t \in \{0,1\}^n$ and a vector $\vv \in C$, it outputs $1$ if and only if $f^H(\vv) = G(t)$ and $t \not\in S$.
\end{itemize}

A quantum algorithm $\As$, with the same advice as in \cite{EC:Liu23}, achieves the following:
\begin{itemize}
    \item The algorithm $\As$ on oracles $G, \ora$ (which will be either $\ora[G, H, \emptyset]$ or $\ora[G, H, S]$), uniformly at random samples $t \in \{0,1\}^n$. It then uses the quantum advice to compute a vector $\vv$ such that $f^H(\vv) = G(t)$ and outputs $\ora(t, \vv)$. 
    \item When $\ora[G, H, \emptyset]$ is given, by the correctness of the YZ algorithm, $\As$ will output $1$ with an overwhelming probability. 
    \item When $\ora[G, H, S]$ is given, by the definition and the condition $|S| \geq 2/3 \cdot 2^n$, $\As$ will output $1$ with a probability at most $1/3$. 
\end{itemize}

On the other hand, for any quantum algorithm $\Bs$ with classical queries and bounded-sized classical proof, it can not distinguish between the case of having access to $G, \ora[G, H, \emptyset]$ or $G, \ora[G, H, S]$. On a high level, the only way to tell the difference is by finding an input $(t, \vv)$ such that $\ora[G, H, \emptyset](t, \vv) \ne \ora[G, H, S](t, \vv)$. This will require  $\Bs$ to query on an input $(t, \vv)$ such that $f^H(\vv) = G(t)$, which is difficult for $\Bs$ with only classical advice. 

Finally, we mount our new separation result on the diagonalization argument and construct a language that shows the separation between {\sf QMA} and {\sf QCMA} relative to a classically accessible classical oracle. Please refer to \Cref{sec:QMA_QCMA_classical} for full details. 

\paragraph{Separation between {\sf QMA} and {\sf QCMA} relative to a distributional oracle.}
To separate $\QMA$ from $\QCMA$, we would try to separate two distributions of oracles, namely the {\sf YES} and {\sf NO} distributions, defined below: For a random function $H\colon[n]\times \Sigma\to \bit$, 
\begin{itemize}
    \item If it is a {\sf YES} instance, an oracle distribution $\{\ora[r]\}_{r}$ with the index being drawn uniformly at random, such that: $\ora[r]$ takes $\vv\in\Sigma^n$ and evaluates $f^H_C(\vv)$, and outputs 1 iff $f^H_C(\vv)=r$; 
    \item If it is a {\sf NO} instance, the oracle always outputs $0$. 
\end{itemize}


To see that the two distributions can be distinguished using a $\QMA$ machine, notice that we can take the quantum advice as in~\cite{EC:Liu23}, and generate $\vv$ such that $f^H_C(\vv)=r$. By querying $\ora$ at $\vv$, we can distinguish whether the oracle belongs to the {\sf YES} distribution or the {\sf NO} distribution.

To prove the two distributions cannot be distinguished by any $\QCMA$ machine, we notice that the difference between {\sf YES} and {\sf NO} oracles is only on inputs $\vv$ that $f^H_C(\vv)=r$. Therefore, we reduce it to the hardness of finding $\vv$ for random $r$ such that $f^H_C(\vv)=r$, even when given classical advice. 

In~\Cref{sec:QMA_QCMA}, we define unary languages $\calL_i$ and their related oracle distributions for each $n$. By a standard diagonalization argument, we can argue that there exists some language $\calL$ that is in $\QMA^\ora$ but not in $\QCMA^\ora$.
\subsection{Concurrent Work}\label{sec:concurrent}
A concurrent work by Aaronson, Buhrman, and Kretschmer \cite{ABK23}, among many results, proves $\mathsf{FBQP}/\poly\neq \mathsf{FBQP}/\qpoly$ \emph{unconditionally} where $\mathsf{FBQP}/\poly$ and $\mathsf{FBQP}/\qpoly$ are the relational variants of $\mathsf{BQP}/\poly$ and $\mathsf{BQP}/\qpoly$, respectively. The key insight behind the result is an observation that a variant of hidden matching problem~\cite{STOC:BarJayKer04} gives an exponential separation between classical and quantum one-way communication complexity with short input length for Bob. They essentially prove that any such a separation with efficient Bob can be used to prove $\mathsf{FBQP}/\poly\neq \mathsf{FBQP}/\qpoly$. 
We independently had an observation that \cite{FOCS:YamZha22} gives such a separation of classical and quantum one-way communication complexity with short input length for Bob (see \Cref{sec:communication}). 
By relying on their proof, which is fairly easy in hindsight,  it seems possible to prove $\mathsf{FBQP}/\poly\neq \mathsf{FBQP}/\qpoly$ by using \cite{FOCS:YamZha22} instead of the hidden matching problem.

A crucial difference between the one-way communication variant of \cite{FOCS:YamZha22} and the hidden matching problem is that the hardness of the former with classical communication holds even if Bob can classically query Alice's input. Due to this difference, one cannot reprove  $\BQP/\poly\neq \BQP/\qpoly$  and $\QMA\neq \QCMA$ 
relative to a classically-accessible classical oracle by using the hidden matching problem instead of \cite{FOCS:YamZha22}. On the other hand, it seems possible to prove $\QMA\neq \QCMA$ relative to a distributional quantumly-accessible classical oracle by using (a parallel repetition variant of) their variant of the hidden matching problem instead of \cite{FOCS:YamZha22}.
\section{Preliminaries}
\paragraph{Basic notations.}
For a set $X$, 
$|X|$ denotes the cardinality of $X$.
We write $x\la X$ to mean that $x$ is uniformly taken from $X$.
For a distribution $\mathcal{D}$ over classical strings, we write $x\la \mathcal{D}$ to mean that $x$ is sampled from the distribution $\mathcal{D}$.  
For sets $X$ and $Y$, $\Func(X,Y)$ denotes the set of all functions from $X$ to $Y$. 
For a positive integer $n$, $[n]$ denotes the set $\{1,2,...,n\}$.
QPT stands for ``Quantum Polynomial-Time''.
We use $\poly$ to mean a polynomial and $\negl$ to mean a negligible function.

\paragraph{Oracle variations.}
In the literature on quantum computation, 
when we say that a quantum algorithm has oracle access to $f:X\ra Y$, it usually means that it is given oracle access to an oracle that applies a unitary $\ket{x}\ket{y}\mapsto \ket{x}\ket{y\oplus f(x)}$.
We refer to such a standard oracle as \textbf{quantumly-accessible classical oracles}.
In this paper, we consider the following two types of non-standard oracles. 

The first is \textbf{classically-accessible classical oracles}.
A classically-accessible classical oracle for a classical function $f:X\ra Y$ takes a \emph{classical string} $x\in X$ as input and outputs $f(x)$.
In other words, when an algorithm sends $\sum_{x,y}\alpha_{x,y}\ket{x}\ket{y}$ to the oracle, the oracle first \emph{measures} the first register and then applies the unitary $\ket{x}\ket{y}\mapsto \ket{x}\ket{y\oplus f(x)}$. 
Note that classically-accessible classical oracles apply non-unitary operations.

The second is \textbf{distributional quantumly-accessible classical oracles}.
They are specified by a distribution $\mathcal{F}$ over classical functions $f$ rather than by a single function $f$. 
When we consider an algorithm that is given oracle access to a distributional quantumly-accessible classical oracle, it works as follows:  
At the beginning of an execution of the algorithm, a function $f$ is chosen according to the distribution $\mathcal{F}$, and then the algorithm has access to a quantumly-accessible classical oracle that computes $f$. Note that $f$ is sampled at the beginning and then the same $f$ is used throughout the execution.

\paragraph{Complexity classes.}
We define the complexity classes which we consider in this paper. 
Specifically, we define $\BQP/\qpoly$, $\BQP/\poly$, $\QMA$, and $\QCMA$ relative to classically-accessible classical oracles and $\QMA$ and $\QCMA$ relative to distributional quantumly-accessible classical oracles.
\begin{definition}[$\BQP/\qpoly$ and $\BQP/\poly$ languages relative to classically-accessible classical oracles.]
    Let $\ora$ be a classically-accessible classical oracle.
    A language $\mathcal{L}\subseteq \bit^*$ belongs to $\BQP/\qpoly$ relative to $\ora$ if 
    there is a QPT machine $\A$ and a polynomial-size family $\{\ket{z_n}\}_{n\in \mathbb{N}}$ of quantum advice such that for any $x\in \bit^*$, 
    \[
    \Pr[\A^\ora(x,\ket{z_{|x|}})=\mathcal{L}(x)]\ge 2/3
    \]
    where $\mathcal{L}(x):=1$ if $x\in \mathcal{L}$ and otherwise $\mathcal{L}(x):=0$.
    
$\BQP/\poly$ is defined similarly except that the advice is required to be classical.
\end{definition}

\begin{definition}[$\QMA$ and $\QCMA$ languages relative to classically-accessible classical oracles.]
    Let $\ora$ be a classically-accessible classical oracle.
    A language $\mathcal{L}\subseteq \bit^*$ belongs to $\QMA$ relative to $\ora$ if 
    there is a QPT machine $V$ with classical access to its oracle and a polynomial $p$ such that
    the following hold:
    \begin{description}
    \item[Completeness]
    For any $x\in \mathcal{L}$, there is a 
    $p(|x|)$-qubit witness $\ket{w}$ such that 
    \[
    \Pr[V^\ora(x,\ket{w})=1]\ge 2/3.
    \]
    \item[Soundness]
      For any $x\notin \mathcal{L}$ and 
    $p(|x|)$-qubit witness $\ket{w}$, 
    \[
    \Pr[V^\ora(x,\ket{w})=1]\le 1/3.
    \]
    \end{description}
$\QCMA$ is defined similarly except that the witness is required to be classical.
\end{definition}

\begin{definition}[$\QMA$ and $\QCMA$ languages relative to distributional quantumly-accessible classical oracles.]
    Let $\mathcal{F}$ be a distributional quantumly-accessible classical oracle, i.e., it specifies a distribution over classical functions $f$.
    A language $\mathcal{L}\subseteq \bit^*$ belongs to $\QMA$ relative to $\mathcal{F}$ if 
    there is a QPT machine $V$ with quantum access to its oracle and a polynomial $p$ such that
    the following hold:
    \begin{description}
    \item[Completeness]
    For any $x\in \mathcal{L}$, there is a 
    $p(|x|)$-qubit witness $\ket{w}$ such that 
    \[
    \Pr_{f\la \mathcal{F}}[V^f(x,\ket{w})=1]\ge 2/3.
    \]
    \item[Soundness]
      For any $x\notin \mathcal{L}$ and 
    $p(|x|)$-qubit witness $\ket{w}$, 
    \[
    \Pr_{f\la \mathcal{F}}[V^f(x,\ket{w})=1]\le 1/3.
    \]
    \end{description}
$\QCMA$ is defined similarly except that the witness is required to be classical.
\end{definition}
\begin{remark}
Notice that the quantum/classical witness $\ket{w}/w$ can only depend on the distribution $\F$ rather than a specific oracle $f$. 
\end{remark}


\paragraph{Non-Uniformity in the ROM.} Prior work has developed a number of tools to characterize the power of a non-uniform adversary in the random oracle model (ROM). Note that we are considering the ROM where we only allow classical access to random oracles unlike the quantum ROM~\cite{AC:BDFLSZ11}. This is sufficient for our purpose because we use these tools only for proving $\BQP/\qpoly \neq \BQP/\poly$ and $\QMA \neq \QCMA$ relative to \emph{classically-accessible} classical oracles. 

Similar to \cite{EC:Liu23}, we will be using the presampling technique, introduced by \cite{C:Unruh07} and further developed by \cite{EC:CDGS18,TCC:GLLZ21}. 
First, we define games in the ROM.
\begin{definition}[Games in the ROM]
\label{def:game}
A game $\game$ 
in the ROM is specified by three classical algorithms $\samp^H, \query^H$, and $\ver^H$ where $H\la \Func(X,Y)$ 
for some sets $X,Y$: 
\begin{itemize}
    \item $\samp^H(r)$: it is a deterministic algorithm that takes uniformly random coins $r \in \mathcal{R}$ as input, and outputs a challenge $\ch$. \footnote{We note that the randomness $r$ is separate from the randomness of $H$. $r$ is used to sample the challenge.}
    \item $\query^H(r, \cdot)$: it is a deterministic classical algorithm that hardcodes the randomness $r$ used to construct the challenge and provides the adversary's online queries. \footnote{As an example, for most applications, $\query^H(r, \cdot)=H(\cdot)$.}
    \item $\ver^H(r, \ans)$: it is a deterministic algorithm that takes the same random coins for generating a challenge and an alleged answer $\ans$, and outputs $b$ indicating whether the game is won ($b = 1$ for winning). 
\end{itemize}
Let $T_\samp$ be the number of queries made by $\samp$ and $T_\ver$ be the number of queries made by $\ver$. 

\medskip

For a fixed $H\in \Func(X,Y)$ and a quantum algorithm $\As$, the game $\game^H_{\As}$ 
is executed as follows:
\begin{itemize}
    \item A challenger $\Cs$ samples $\ch \gets \samp^H(r)$ using uniformly random coins $r$.
    \item A (uniform or non-uniform) quantum algorithm $\As$, that has classical oracle access 
    to $\query^H(r, \cdot)$,  
    takes $\ch$ as input and outputs $\ans$. We call $\As$ an online adversary/algorithm. 
    \item $b \gets \ver^H(r, \ans)$ is the game's outcome. 
\end{itemize}
\end{definition}
\begin{definition}\label{def:game_ROM}
    We say that a game $\game$ in the ROM has security $\delta(Z,Q):=\delta$ if
    $$
    \max_{\A}\Pr_H[\game_\A^H=1]\le \delta
    $$
    where 
    $H\la\Func(X,Y)$ and 
    $\max$ is taken over all $\A$ with $Z$-bit classical advice that makes $Q$ classical queries. 
\end{definition}

The presampling technique relates the probability of success of a non-uniform algorithm with classical queries to a random oracle with the success probability of a uniform algorithm in the $P$ bit-fixing game, as defined below.

\begin{definition}[Games in the $P$-BF-ROM]
\label{def:pbfrom} 
A game $\game$ in the $P$-BF-ROM is specified by two classical algorithms $\samp^H$ and $\ver^H$ that work similarly to those in \Cref{def:game}. 

For a fixed $H\in\Func(X,Y)$  and a pair of algorithms $(f,\As)$, the game $\game^H_{f,\As}$ 
is executed as follows:
\begin{itemize}
    \item {\bf Offline Stage:} Before a game starts, an offline algorithm $f$ (having no input) generates a list $\mathcal{L} = \{(x_i, y_i)\}_{i \in [P]}\in (X\times Y)^P$ containing at most $P$ input-output pairs (all $x_i$'s are distinct).
    \item {\bf Online Stage:} The game is then executed with $\A$ that knows $\mathcal{L}$ and has oracle access to $H$ similarly to \Cref{def:game}. $H$ is a function drawn at random such that it satisfies $\mathcal{L}$.
\end{itemize}
\end{definition}
\begin{definition}
    We say that a game $\game$ in the $P$-BF-ROM has security $\nu(P,Q):=\nu$ if  
    $$
    \max_{f,\A}\Pr_H[\game_{f,\A}^H=1]\le \nu
    $$
    where 
    $H\la\Func(X,Y)$ and 
    $\max$ is taken over all 
    $f$ that outputs $P$ input-output pairs and 
    $\A$ that makes $Q$ classical queries.  
\end{definition}

\begin{theorem}[{\cite[Theorem A.1]{EC:Liu23}}\footnote{Similar theorems with slightly worse bounds are presented in \cite[Theorem 5 and 6]{EC:CDGS18} and \cite[Theorem 3]{TCC:GLLZ21}.}]
\label{thm:classical_bf_to_nonuniform}
Let $\game$ be any game with $Q_\samp, Q_\ver$ being the number of queries made by $\samp$ and $\ver$. 
For any classical advice length $Z$, and number of online queries $Q$: 
\begin{enumerate}
\item \label{item:search} 
For $P=Z(Q+Q_\samp+Q_\ver)$, 
if $\game$ has security $\nu(P, Q)$ in the $P$-BF-ROM, then it has security $\delta(Z, Q) \leq 2 \cdot \nu(P, Q)$ against non-uniform unbounded-time algorithms with $Z$ bits of classical advice and $Q$ classical queries.  
\item \label{item:decision} 
For any $P > 0$, 
if $\game$ has security $1/2+\nu(P, Q)$ in the $P$-BF-ROM, then it has security $\delta(Z, Q) \leq 1/2 + \nu(P, Q) + {Z(Q + Q_\ver + Q_\samp)}/{P}$ against non-uniform unbounded-time algorithms with $Z$ bits of classical advice and $Q$ classical queries. 
\end{enumerate}
\end{theorem}

We can use the above correspondence to bound the success probability of a non-uniform quantum algorithm with classical queries to the Yao's Box game \cite{STOC:Yao90}.


\begin{lemma}[Yao's Box with Classical Queries~\cite{STOC:Yao90, C:DeTreTul10}]
\label{lem:classical-Yao-box}
    Let $G:[N] \to \{0, 1\}$ be a random function. Let $\mathcal{A}$ be an unbounded-time algorithm, with $Z$ bits of (classical) advice $z_G$ and $Q$ classical queries to $G$. 
    The probability that $\mathcal{A}$ computes $G(x)$ without querying $G$ at random index $x$ is at most
    \[
        \Pr_x[\mathcal{A}^G(z_G, x) = G(x)] \leq \frac{1}{2} + 2 \sqrt{\frac{Z(Q + 1)}{N}}.
    \]
\end{lemma}

The above lemma was essentially proved in~\cite{C:DeTreTul10}, but we offer here an alternative proof using the presampling technique of~\Cref{thm:classical_bf_to_nonuniform}.


\begin{proof}[Proof of \Cref{lem:classical-Yao-box}]
    We consider the bit-fixing game in the $P$-BF-ROM, where we fix the value of $G$ at $P$ positions in the offline phase. Since $\mathcal{A}$ is not allowed to query $G$ at position $x$ even in the $P$-BF-ROM (because it queries $G$ via the original $\query^G$), the only way for $\mathcal{A}$ to successfully compute $G(x)$ is if $x$ is included in the set of fixed positions during the offline phase. This happens with probability $\frac{P}{N}$ and thus
    $$\nu(P, T) \leq \frac{P}{N}.$$

    In this game $Q_\samp = 0$, since the Sampler outputs a challenge $x \in [N]$ uniformly at random, without the need to perform any queries. The Verifier only needs to query $G$ at position $x$, thus $Q_\ver = 1$. The statement of~\Cref{item:decision} of \Cref{thm:classical_bf_to_nonuniform} with $P = \sqrt{Z(Q+1)N}$ implies that the advantage of a non-uniform algorithm with advice $z_G$ and $Q$ queries is at most
    $$\delta(Z, Q) \leq \frac{1}{2} + 2 \sqrt{\frac{Z(Q+1)}{N}}$$.
\end{proof}

The series of definitions and lemmas that follow relate the ROM with auxiliary input to two related kinds of distributions: bit-fixing and dense distributions. Looking ahead, these are used in the separation between $\QMA$ and $\QCMA$ relative to a classically-accessible oracle in \Cref{sec:QMA_QCMA_classical}.
\begin{definition}  
    For $P\in \mathbb{N}$, 
    we say that a distribution $\calD$ supported by functions $G\colon X\to Y$ is 
    a $P$-bit-fixing distribution if there is a subset $S\subseteq X$ such that 
    $|S|\le P$ and 
    \begin{itemize}
        \item for all  $x\in S$,  $G(x)$ takes the same value for all $G$ in the support of $\calD$, and 
        \item for all  $x\notin S$,  $G(x)$ is uniformly distributed over $Y$ when we sample $G\la \calD$. 
    \end{itemize}
    We say that $\calD$ is fixed on $x$ if and only if $x\in S$. 
\end{definition}

In addition, we will consider the following generalization of a $P$-bit-fixing distribution. A $(P, 1-\delta)$-dense distribution is one where the non-fixed coordinates are close to uniformly distributed. The parameter $1 - \delta$ controls how close to uniform they are.
\begin{definition}
    A distribution $\calD$ supported on functions $G:X \to Y$ is called
    \begin{itemize}
        \item $(1-\delta)$-dense if for every subset $I \subset X$,
        $$H_\infty(G(I)) \geq (1 - \delta)\cdot |I|\cdot \log |Y| = (1-\delta)\cdot |Y|^{|I|}.$$
        \item $(P, 1-\delta)$-dense if it is fixed on at most $P$ coordinates and is $(1-\delta)$-dense on the rest.
    \end{itemize} 
\end{definition}

\angelos{added formal definition of corresponding bit-fixing distribution}
We will further define for each $(P, 1-\delta)$-dense distribution $\calD$ its corresponding $P$-bit-fixing distribution to be the unique bit-fixing distribution with the same fixed coordinates. Formally:

\begin{definition}
    Let $\calD$ be a $(P, 1-\delta)$-dense distribution supported on functions $G:X \to Y$. Its corresponding $P$-bit-fixing distribution $\calD'$ is also supported on functions $G:X \to Y$ and is fixed on the same coordinates that $\calD$ is, and uniform on the rest.
\end{definition}

\begin{claim}{\cite[Claim 2]{EC:CDGS18}}
    \label{cl:CDGSalmostbitfixing}
    Let $G:X \to Y$ be a random oracle, and let $f$  be an arbitrary function that maps a random oracle to $S$-bit advice. For an arbitrary classical advice $z \in \{0, 1\}^S$, let $G_z$ be the distribution of $G$ conditioned on $z = f(G)$. 
    Let $S_z = |X|\log |Y| - H_\infty(G_z)$ be the min-entropy deficiency of $G_z$. For every $\gamma, \delta > 0$, $G_z$ is $\gamma$-close to a convex combination of finitely many $(P', 1-\delta)$-dense sources for
    $$P' = \frac{S_z + \log 1/\gamma}{\delta\cdot \log |Y|}.$$
\end{claim}

In \Cref{cl:CDGSalmostbitfixing}, the number of fixed entries is related to $S_z$, the min-entropy deficiency of $G_z$, which can be exponentially large. The following claim shows that $S_z$ is close to the length $S$ of the classical advice with high probability.

\begin{claim}{\cite[Claim 4]{EC:CDGS18}}
    \label{cl:CDGS-length-of-fixed}
    Let $G:X \to Y$ be a random oracle, and let $f$ be an arbitrary function that maps a random oracle to some $S$-bit advice. Additionally, let $S_z$ be the min-entropy deficiency of $G_z$ as defined in \Cref{cl:CDGSalmostbitfixing}. Then the distribution of $S_z$ satisfies
    $$\Ex_{G,z \la f(G)}[S_z] \leq S\quad \text{and}\quad \Pr_{G,z\la f(G)}[S_z > S + \log 1/\gamma] \leq \gamma.$$ 
\end{claim}

Finally, we will use the following claim from \cite{EC:CDGS18}, which provides a way to translate between a $(P, 1-\delta)$-dense distribution to a $P$-bit-fixing one.

\begin{claim}{\cite[Claim 3]{EC:CDGS18}}
\label{cl:dense-to-bit-fixing}
    Let $\calD$ be a $(P, 1-\delta)$-dense distribution supported on functions $G:X \to Y$ and $\calD'$ be its corresponding $P$-bit-fixing distribution. Then, for any (adaptive) distinguisher $\B$ that makes at most $T$ oracle queries,
    $$\left|\Pr[\B^\calD = 1] - \Pr[\B^{\calD'} = 1]\right| \leq T\delta \cdot \log |Y|.$$ 
\end{claim}

\paragraph{One-way to hiding lemma.} 
We review a lemma called the one-way to hiding lemma originally proven by Unruh \cite{JACM:Unruh15}. The following is a variant proven in \cite{C:AmbHamUnr19}. 
\begin{lemma}[{One-Way to Hiding Lemma~\cite[Theorem~3]{C:AmbHamUnr19}}]
\label{lem:o2h}
Let $S \subseteq \mathcal{X}$ be random.
Let $G,H \colon \mathcal{X} \to \mathcal{Y}$ be random functions satisfying $\forall x \not\in S~[G(x) = H(x)]$.
Let $z$ be a random classical bit string or quantum state.
($S,G,H,z$ may have an arbitrary joint distribution.)
Let $\A$ be an oracle-aided quantum algorithm that makes at most $Q$ quantum queries. 
Let $\B$ be an algorithm that on input $z$ chooses $i\la [q]$, runs $\A^H(z)$, measures $\A$'s $i$-th query, and outputs the measurement outcome.
Then we have
\[
|\Pr[\A^{G}(z)=1]-\Pr[\A^H(z)=1]|\leq 2Q\sqrt{\Pr[\B^{H}(z)\in S]}.
\]
\end{lemma}
We remark that we consider quantum access to the oracles in the above lemma since we use it in the proof of the separation between $\QMA$ and $\QCMA$ relative to a distributional \emph{quantumly-accessible} oracle in \Cref{sec:QMA_QCMA}.
\section{\texorpdfstring{Quantum vs Classical Advice for \cite{FOCS:YamZha22}}{Classical Advice for the YZ problem}}\label{sec:YZ}

We review the result of \cite{FOCS:YamZha22}, which gives an $\NP$-search problem that is easy for $\BQP$ machines but hard for $\BPP$ machines relative to a random oracle. 
Then we show that the problem is easy for $\BQP$ machines with quantum advice and no online query but hard for unbounded-time machine with polynomial-size classical advice and polynomially-many classical online queries relying on an observation of \cite{EC:Liu23}. 
\begin{definition}[\cite{FOCS:YamZha22}]\label{def_YZ_function}
Let $C\subseteq \Sigma^n$ be a code over an alphabet $\Sigma$ and $H: [n]\times \Sigma \rightarrow \bit$ be a function. The following function is called a YZ function with respect to $C$ and $H$:
\begin{align*}
    &f_C^H:C\rightarrow \bit^n\\
&f_C^H(v_1,v_2,...,v_n)=H(1,v_1)||H(2,v_2)||...||H(n,v_n).
\end{align*}
\end{definition}
\begin{remark}\label{rem:subscript_n}
When we refer to a code $C\subseteq \Sigma^n$, it actually means a family $\{C_n\}_{n\in \mathbb{N}}$ of codes $C_n\subseteq \Sigma_n^n$ over an alphabet $\Sigma_n$.  
We often omit the dependence of $n$ for notational simplicity.
\end{remark}

\cite{FOCS:YamZha22} shows that there is an error correcting code $C$, which satisfies a property called \emph{list-recoverability}, such that $f_C^H$ is easy to invert with quantum queries to $H$ but hard to invert with classical queries to $H$ where $H$ is modeled as a random oracle.   Liu~\cite{EC:Liu23} observed that the quantum inversion algorithm of \cite{FOCS:YamZha22} does not need to make adaptive queries, and having a polynomial-size quantum advice on $H$ that is independent of the target $r$ suffices. In addition, he proved that classical advice does not suffice for the task if no adaptive query is allowed. Combining the above, we have the following theorem.
\begin{theorem}[\cite{FOCS:YamZha22,EC:Liu23}]\label{thm:YZLiu_separation}
There is a code $C\subseteq \Sigma^n$ satisfying the following:
\begin{enumerate}
    \item\label{item:YZ_easiness_with_quantum_advice} {\bf(Easiness with Quantum Advice)}
    There is a QPT algorithm $\A$ and a family of $\poly(n)$-qubit quantum advice $\{\ket{z_H}\}_{H}$
    such that for any $r\in\bit^n$, 
    \begin{align*}
        \Pr_{H}[
    f_C^H(\A(\ket{z_H},r))=r)
    ]\ge  1-\negl(n)
    \end{align*}
where $H\la \Func([n]\times \Sigma,\bit)$.
\item {\bf(Hardness with Classical Advice)}  For any unbounded-time algorithm $\B$ and polynomial $s$, there is a negligible function $\mu$ such that for any family of $s(n)$-bit classical advice $\{z_H\}_{H}$,  
     \begin{align*}
        \Pr_{H,r}[
    f_C^H(\B(z_H,r))=r
    ]
    \le \mu(n)
    \end{align*}
where $H\la \Func([n]\times \Sigma, \bit)$
and $r\la \bit^n$. 
\end{enumerate}
Moreover, $C$ satisfies $(\zeta,\ell,L)$-list-recoverability, i.e., for any subset $T_i\subseteq \Sigma$ such that $|T_i|\le \ell$ for every $i\in[n]$, 
\begin{align*}
    |\{(v_1,...,v_n)\in C:|\{i\in [n]:x_i\in T_i\}|\ge (1-\zeta)n\}|\le L
\end{align*}
where $\zeta=\Omega(1)$, $\ell=2^{n^c}$, and $L=2^{\tilde{O}(n^{c'})}$ for some constants $0<c<c'<1$.
\end{theorem}

We extend \Cref{thm:YZLiu_separation} to require that the hardness with classical advice holds even if $\A$ is given a classical access to $H$. This can be seen as a unification of \cite{FOCS:YamZha22} and \cite{EC:Liu23}.

\begin{theorem}\label{thm:YZLiu_separation_classical_query}
Let $C\subseteq \Sigma^n$ be the code in \Cref{thm:YZLiu_separation}, the following holds:
\begin{description}
\item[Hardness with Classical Advice and Classical Queries]  For any unbounded-time algorithm $\B$ that makes polynomially many classical queries to $H$ and polynomial $s$, there is a negligible function $\mu$ such that for any family of $s(n)$-bit classical advice $\{z_H\}_{H}$,  
     \begin{align*}
        \Pr_{H,r}[
    f_C^H(\B^H(z_H,r))=r
    ]
    \le \mu(n)
    \end{align*}
where $H\la \Func([n]\times \Sigma, \bit)$
and $r\la \bit^n$. 
\end{description} 
\end{theorem}
\begin{proof}
    The proof is obtained by combining the proofs of \cite{FOCS:YamZha22} and \cite{EC:Liu23}. We give a full proof for completeness.
    By \Cref{item:search} of \Cref{thm:classical_bf_to_nonuniform}, we only have to prove that the winning probability of the following game in the $P$-BF-ROM is $\negl(n)$ for any $P=\poly(n)$ and an unbounded-time algorithm $\B$ that makes $Q=\poly(n)$ online classical queries.
\begin{description}
    \item[Offline Stage] 
    $\B$ chooses list $L=\{x_k,y_k\}_{k \in [P]}$ of $P$ input-output pairs 
    where 
    $x_k\in [n]\times \Sigma$ and $y_k\in \bit$ for each $k\in[P]$
    and $x_k\neq x_{k'}$ for all $k\neq k'$. 
    Then the random oracle $H\la \Func([n]\times \Sigma,\bit)$ is uniformly chosen under the constraint that $H(x_k)=y_k$ for all $i\in[P]$.
    \item[Online Stage]
    $\B$ takes $r\la \bit^n$ as input, 
    makes $Q$ classical queries to $H$, and outputs $\vv^*=(v_1^*,...,v_n^*)\in \Sigma^n$.
    \item[Decision] 
    $\B$ wins if $\vv^*\in C$ and $f_C^H(\vv^*)=r$.
\end{description}
For $i\in[n]$, let $T_i:=\{v_i\in \Sigma: x_k=(i,v_i)~\text{for some}~k\in [P]\}$. 
Let 
\begin{align*}
    \mathsf{Good}:=\{(v_1,...,v_n)\in C:|\{i\in [n]:x_i\in T_i\}|\ge (1-\zeta)n\}.
\end{align*}
By the $(\zeta,\ell,L)$-list-recoverability and $|T_i|\le P\le 
2^{n^{c}}
\le 
\ell$ (for sufficiently large $n$), we have 
$|\mathsf{Good}|\le L$.  
We consider the following two cases:
\paragraph{Case 1:~$\vv^*\in \mathsf{Good}$.} 
For each element  $\vv\in \mathsf{Good}$, $\Pr_{r\la \bit^n}[f_C^H(\vv)=r]=2^{-n}$. Thus, by the union bound, the probability that $\B$ wins is at most $|\mathsf{Good}|\cdot 2^{-n}=2^{-\Omega(n)}$ by $|\mathsf{Good}|\le L=2^{\tilde{O}(n^{c'})}$ and $c'<1$.

\paragraph{Case 2:~$\vv^*\notin \mathsf{Good}$.}  
The analysis of this case is very similar to the proof of soundness in \cite{FOCS:YamZha22} and the following proof is partially taken verbatim from theirs. 
For each $i\in[n]$ and $j\in [Q]$, let $S_i^j \subseteq \Sigma$ be the set of elements $v_i$ such that $\B$ ever queried  $(i,v_i)$ by the point when it has just made the $j$-th query. 
Let $\hat{S}_i^j:= S_i^j \cup T_i$ for each $i\in[n]$ and $j\in \{0,1,...,Q\}$.
Without loss of generality, we assume that $v_i^*\in \hat{S}_i^Q$ for all $i\in[n]$.\footnote{This might increase $Q$ by at most $n$, but $T$ is still $\poly(\secp)$ anyway.}
After the $j$-th query, we say that a codeword $\vv=(v_1,...,v_n)\in C$ is \emph{$K$-queried} 
if there is a subset $I\in [n]$ such that $|I|= K$, $v_i\in \hat{S}_i^j$ for all $i\in I$, and $v_i\notin \hat{S}_i^j$ for all $i\notin I$.
By the assumption that $\vv^*\notin \mathsf{Good}$, $\vv^*$ is $K_0$-queried for some $K_0\le \lceil(1-\zeta)n \rceil$ right after the offline stage.
By the assumption that $v_i^*\in \hat{S}_i^Q$ for all $i\in[n]$, $\vv^*$ is $n$-queried after the $Q$-th query. 
Since a $K$-queried codeword either becomes $(K+1)$-queried or remains $K$-queried by a single query, $\vv^*$ must be $K$-queried at some point of the execution of $\B$ for all $K=K_0,K_0+1,...,n$.

We consider the number of codewords that ever become $K$-queried for $K=\lceil(1-\zeta)n \rceil$. 
If $\vv=(v_1,...,v_n)\in C$ is $\lceil(1-\zeta)n \rceil$-queried at some point, the number of $i$ such that $v_i\in \hat{S}_i^{Q}$ is at least $\lceil(1-\zeta)n \rceil$ since  $\hat{S}_i^j\subseteq \hat{S}_i^{Q}$ for all $i,j$. We have $|\hat{S}_i^{Q}|\le Q+P=\poly(n)<2^{n^c}$ for sufficiently large $n$. 
On the other hand, $C$ is $(\zeta,\ell,L)$-list recoverable for $\ell=2^{n^c}$ and $L=2^{\tilde{O}(n^{c'})}$.  
Thus, the number of  codewords that ever become $\lceil(1-\zeta)n \rceil$-queried is at most $L=2^{\tilde{O}(n^{c'})}$.  

Suppose that we simulate the oracle $H$ for $\B$ via lazy sampling, that is, instead of uniformly choosing random functions at first, we sample  function values whenever they are specified in the list sent in the offline stage or queried in the online stage.  
Suppose that a codeword $\vv$ becomes $\lceil(1-\zeta)n \rceil$-queried at some point of the execution of the experiment. Since the function values on the unqueried $\lfloor\zeta n \rfloor$ positions are not sampled yet, $\vv$ can become a valid proof only if all those values happen to be consistent to $r$, which occurs with probability $\left(\frac{1}{2}\right)^{\lfloor\zeta n \rfloor}=2^{-\Omega(n)}$ by $\zeta=\Omega(1)$. 
Since one of them is the final output $\vv^*$, by the union bound, the probability that $\vv^*$ is a valid proof is at most 
$L\cdot \left(\frac{1}{2}\right)^{\lfloor\zeta n \rfloor}=2^{-\Omega(n)}$ by $L=2^{\tilde{O}(n^{c'})}$ and $c'<1$.
\end{proof}


We have the following corollary. 
The motivation of showing this corollary is to ensure that ``a large fraction of $H$ works for all $r$'' rather than that ``for all $r$, a large fraction of $H$ works''. Looking ahead, this is needed for proving an oracle separation for $\BQP/\poly$ and $\BQP/\qpoly$ (but not for $\QMA$ and $\QCMA$).  
\begin{corollary}\label{cor:YZLiu_separation_classical_query_amp}
  Let $C$ be the code in \Cref{thm:YZLiu_separation}. Then the following hold:
\begin{enumerate}    \item\label{item:amp_quantum_easiness}  {\bf (Easiness with Quantum Advice)}
    There is a QPT algorithm $\A$ and a family of $\poly(n)$-qubit quantum advice $\{\ket{z_\mathcal{H}}\}_{\mathcal{H}}$
    such that 
    \begin{align*}
        \Pr_{\mathcal{H}}[
            \forall r\in \bit^n~
            &\Pr[\exists j\in [n]~\text{s.t.}~ f_C^{H^{(j)}}(\vv^{(j)})=r: (\vv^{(1)},...,\vv^{(n)})\la \A(\ket{z_\mathcal{H}},r)]\ge 1-\negl(n)
        ] \\
        &\ge  1-\negl(n)
    \end{align*}
where $\mathcal{H}=(H^{(1)},...,H^{(n)})\la \left(\Func([n]\times \Sigma,\bit)\right)^n$. 
\item\label{item:amp_classical_hardness}  {\bf (Hardness with Classical Advice and Classical Queries)}  For any unbounded-time algorithm $\B$ that makes $\poly(n)$ classical queries to $\mathcal{H}=(H^{(1)},...,H^{(n)})$ and polynomial $s$, there is a negligible function $\mu$ such that for any family of $s(n)$-bit classical advice $\{z_\mathcal{H}\}_{\mathcal{H}}$,  
     \begin{align*}
        \Pr_{\mathcal{H},r}[
   \exists j\in [n]~\text{s.t.}~ f_C^{H^{(j)}}(\vv^{(j)})=r: (\vv^{(1)},...,\vv^{(n)})\la \B^\mathcal{H}(z_\mathcal{H},r)
    ]
    \le \mu(n)
    \end{align*}
where $\mathcal{H}=(H^{(1)},...,H^{(n)})\la \left(\Func([n]\times \Sigma,\bit)\right)^n$ 
and $r\la \bit^n$. 
\end{enumerate}
\end{corollary}
\begin{proof}
    For proving \Cref{item:amp_quantum_easiness}, we can set the advice as $\ket{z_\mathcal{H}}=\ket{z_{H^{(1)}}}\otimes\ket{z_{H^{(2)}}}\otimes\cdots\otimes\ket{z_{H^{(n)}}}$, and the algorithm just parallel runs the algorithm in~\Cref{thm:YZLiu_separation} for different $H^{(i)}$. 
    Assume the algorithm in~\Cref{thm:YZLiu_separation} satisfies:
    \begin{align*}
         \Pr_{H}[
    f_C^H(\A(\ket{z_H},r))=r)
    ]\ge  1-\eta(n),
    \end{align*}
    for some negligible function $\eta(n)$. For any $r\in\{0,1\}^n$, 
    \begin{align*}
        \Pr_\mathcal{H}[\forall j\in [n], f_C^{H^{(j)}}(\vv^{(j)})\neq r: (\vv^{(1)},...,\vv^{(n)})\la \A(\ket{z_\mathcal{H}},r)]\leq \eta(n)^n.
    \end{align*}
     By an averaging argument, at most $\eta(n)^{n/2}$ fraction of $\mathcal{H}$ will satisfy 
    \begin{align*}
        \Pr[\forall j\in [n], f_C^{H^{(j)}}(\vv^{(j)})\neq r: (\vv^{(1)},...,\vv^{(n)})\la \A(\ket{z_\mathcal{H}},r)]\geq \eta(n)^{n/2}
    \end{align*}
    By applying a union bound over all $r \in \{0, 1\}^n$, we obtain that
    \begin{align*}
        \Pr_\mathcal{H}[\exists r\in\bit^n \Pr[\forall j\in [n], f_C^{H^{(j)}}(\vv^{(j)})\neq r: (\vv^{(1)},...,\vv^{(n)})\la \A(\ket{z_\mathcal{H}},r)]\geq \eta(n)]\leq (4\eta(n))^{\frac{n}{2}}
    \end{align*}
    Applying negation, we obtain the bound above.

    For proving \Cref{item:amp_classical_hardness}, we will show how an adversary $\B$ that breaks hardness with classical advice and classical queries, can be used to construct an adversary $\B'$ that breaks hardness with classical advice and classical queries of~\Cref{thm:YZLiu_separation_classical_query}. To show this, we go through the following steps:
    \begin{enumerate}
        \item For each fixed $j \in [n],\mathcal{H}_{\Bar{j}}=(H^{(1)},\dots,H^{(j-1)},H^{(j+1)},\dots,H^{(n)})$, we define a pair\\ $(\B'[j,\calH_{\Bar{j}}],\{z'_H[j,\calH_{\Bar{j}}]\}_H)$ of an adversary and advice in which $j$ and $\calH_{\Bar{j}}$ is hardwired.
        \item Show that $(\B'[j,\calH_{\Bar{j}}],\{z'_H[j,\calH_{\Bar{j}}]\}_H)$ breaks \Cref{thm:YZLiu_separation_classical_query} on average over the choice of $j$.
        \item Fix the ``best'' $j$ and $\calH_{\Bar{j}}$ (w.r.t. random $H$) to get a fixed pair of algorithm and advice that breaks hardness of \Cref{thm:YZLiu_separation_classical_query}.
    \end{enumerate}
    Specifically, it works as follows. Suppose there exists some adversary $\B$ and some polynomial $q(n)$ such that for sufficiently large $n$,
    \begin{align*}
        \Pr_{\mathcal{H},r}[
   \exists j\in [n]~\text{s.t.}~ f_C^{H^{(j)}}(\vv^{(j)})=r: (\vv^{(1)},...,\vv^{(n)})\la \B^\mathcal{H}(z_\mathcal{H},r)
    ]
    \geq q(n)
    \end{align*}
    
    For each $j$ and $\calH_{\Bar{j}}$, we define $(\B'[j,\calH_{\Bar{j}}],\{z'_H[j,\calH_{\Bar{j}}]\}_H)$ as follows:
    \begin{description}
        \item[$z'_H{[}j,\calH_{\Bar{j}}{]}$:] Given $\calH_{\Bar{j}}$, set
        $$\mathcal{H} = (H^{(1)}, \dots, H^{(j-1)}, H, H^{(j+1)}, \dots, H^{(n)}).$$ Set $z'_H[j,\calH_{\Bar{j}}]$ to the advice of $\B$ for $\mathcal{H}$, i.e. $z'_H[j,\calH_{\Bar{j}}]:=z_\mathcal{H}$.
        \item[${\B'[j,\calH_{\Bar{j}}]}^{H}(z'_H{[}j,\calH_{\Bar{j}}{]},r)$:] 
        It runs $\B^{\mathcal{H}}(z'_H[j,\calH_{\Bar{j}}],r)$ where $\B'$ hardwired $\calH_{\Bar{j}}$ into its algorithm. It simulates the oracle $\mathcal{H}$ for $\B$ by querying its own oracle $H$ as $H^{(j)}$ and simulates other $H^{(i)}(i\neq j)$ by querying the hardwired $\calH_{\Bar{j}}$. 
    \end{description}
    By the above argument, we have 
        \begin{align*}
        &\Pr_{H,r, j, \calH_{\Bar{j}}}[
            f_C^H(\B'[j,\calH_{\Bar{j}}]^H(z_H'[j,\calH_{\Bar{j}}],r))=r
        ] \\
        &= \Pr_{\mathcal{H},r, j}[f_C^{H^{(j)}}(\vv^{(j)})=x: (\vv^{(1)},...,\vv^{(n)})\la \B^\mathcal{H}(z_\mathcal{H},r)
        ] \\
        &\geq \frac{q(n)}{n}.
    \end{align*}
    Thus, by taking $j=j^*$ and $\calH_{\Bar{j^*}}^*$ that makes the above probability the largest for random $\mathcal{H}$, 
    $(\B',\{z'_H\}_H):=(\B'[j^*,\calH_{\Bar{j^*}}^*],\{z'_H[j^*,\calH_{\Bar{j^*}}^*]\}_H)$ breaks
    hardness of \Cref{thm:YZLiu_separation_classical_query}.
    \end{proof}

\section{\texorpdfstring{$\BQP/\poly$ vs $\BQP/\qpoly$ under Classically-Accessible Oracle}{BQP\slash poly vs BQP\slash qpoly under Classically-Accessible Oracle}}\label{sec:advice} 
In this section, we demonstrate a $\BQP/\qpoly$ and $\BQP/\poly$ separation relative to a classically-accessible classical oracle.
 

The main technical lemma needed for proving the separation
is the following.
\begin{lemma}\label{lem:separation_poly_qpoly_dist}
    There is a family of distributions $\{\dist_n\}_{n\in \mathbb{N}}$, where $\dist_n$ is supported on tuples $(G,\ora)$ of functions $G:\bit^n\rightarrow \bit$ and $\ora: \bit^{p(n)}\rightarrow \bit^{q(n)}$ for some polynomials $p$ and $q$, 
    satisfying the following:
    \begin{enumerate}
  \item\label{item:separation_poly_qpoly_dist_quantum_easiness} 
 {\bf (Easiness with Quantum Advice.)}
    There is a QPT algorithm $\A$ with classical access to $\ora$ and a family of $\poly(n)$-qubit quantum advice $\{\ket{z_\ora}\}_{\ora}$
    such that  
    \begin{align*}
        \Pr_{(G,\ora)\la \dist_n}[
    \forall x\in \bit^n~
    \Pr[\A^{\ora}(\ket{z_\ora},x)=G(x)]\ge 1-\negl(n)
    ]\ge  1-\negl(n).
    \end{align*}
    \item\label{item:separation_poly_qpoly_dist_classical_hardness} {\bf (Hardness with Classical Advice.)} 
     For any unbounded-time algorithm $\B$ that makes $\poly(\secp)$ classical queries to $\ora$ and a family of $\poly(n)$-bit classical advice $\{z_\ora\}_{\ora}$,  
     \begin{align*}
        \Pr_{\substack{(G,\ora)\la \dist_n \\x\la \bit^n}}[
    \B^{\ora}(z_\ora,x)=G(x) 
    ]
    \le \frac{3}{5}
    \end{align*}
for all sufficiently large $n$.
\end{enumerate}
\end{lemma}

For proving \Cref{lem:separation_poly_qpoly_dist}, we prepare the following lemma.
\begin{lemma}\label{lem:query_amplitude}
    Let $G:\bit^n\rightarrow \bit$ be a uniformly random function.  For an unbounded-time algorithm $\A$ that makes $\poly(\secp)$ classical queries to $G$ and a family of $\poly(n)$-bit classical advice $\{z_G\}_{G}$,  suppose that the following holds:
    \begin{align*}
        \Pr_{G,x\la \bit^n}[\A^{G}(z_G,x)=G(x)]> \frac{3}{5}.
    \end{align*} 
    Then the probability that $x$ is contained in the query list is at least $\frac{1}{20}$ for $\frac{1}{30}$ fraction of $x \in \bit^n$ for sufficiently large $n$. 
    
\end{lemma}
\begin{proof}
    For each $x\in\bit^n$, we define $G_x$ as the random function $G$ with its input on $x$ removed, i.e. $G_x(x') = G(x')$ 
     for $x' \neq x$ and $G_x(x) = 0$. Since $\A$ only makes classical queries to $G$, the only way for it to distinguish $G$ from $G_x$ is to query the oracle at $x$. Denote by $\delta_x$ the probability that $x$ is in the query list of $\A^{G_x}$, where the probability is over the randomness of $\A$ and $G_x$. 
    We obtain that, 
    \begin{align*}
        |\Pr_{G}[\A^{G}(z_G,x)=G(x)]-
        \Pr_{G}[\A^{G_x}(z_G,x)=G(x)]|\le \delta_x.
    \end{align*}
    Now we consider the case when we uniform randomly choose $x\la\bit^n$, and require $\A^{G_x}(z_G,x)$ to output $G(x)$. This is exactly Yao's box problem, where the adversary is required to output $G(x)$ without querying $x$. By~\Cref{lem:classical-Yao-box}, we have the following bound for Yao's box with classical queries and classical advice:
    \begin{align*}
        \Pr_{G,x}[\A^{G_x}(z_G,x)=G(x)]\leq \frac{1}{2}+2\sqrt{\frac{|z_G|(Q+1)}{2^n}}=\frac{1}{2}+\negl(n)
    \end{align*}
    where we assume that $\A$ makes $Q$ queries.
    Thus we have that 
    \begin{align*}
        \Pr_{G,x}[\A^{G}(z_G,x)=G(x)]- \Pr_{G,x}[\A^{G_x}(z_G,x)=G(x)]\geq \frac{1}{10}-\negl(n),
    \end{align*}
    Therefore, we have
    \begin{align*}
        \Ex_x [\delta_x]\geq \Ex_x \left[\left|\Pr_{G}[\A^{G}(z_G,x)=G(x)]-
        \Pr_{G}[\A^{G_x}(z_G,x)=G(x)]\right|\right]\geq \frac{1}{10}-\negl(n)
    \end{align*}

    We now show that $\delta_x$ is at least $\frac{1}{20}$ with probability $\frac{1}{30}$ for sufficiently large $n$.
    $$\Pr_x\left[\delta_x \geq \frac{1}{20}\right] + \left(1 - \Pr_x\left[\delta_x \geq \frac{1}{20}\right]\right)\cdot \frac{1}{20} \geq \mathbb{E}_x[\delta_x] \geq \frac{1}{10} - \negl(n)$$
    $$\implies \Pr_x\left[\delta_x \geq \frac{1}{20}\right] \geq \frac{1}{20} - \negl(n)$$ 
    Thus for sufficiently large $n$, for a $\frac{1}{20}-\negl(n) \geq \frac{1}{30}$ fraction of $x \in \{0, 1\}^n$, $\A$ will query $x$ with probability at least $\frac{1}{20}$.
\end{proof}

Then we prove \Cref{lem:separation_poly_qpoly_dist}.
\begin{proof}[Proof of \Cref{lem:separation_poly_qpoly_dist}]
We define $\dist_n$ to be the distribution that samples $G$ and $\ora$  as follows:
\begin{description}
    \item[$\dist_n$:] 
    Let  $C\subseteq \Sigma^n$ be the code in \Cref{cor:YZLiu_separation_classical_query_amp}. 
  It samples random functions $G:\bit^n\rightarrow \bit$ and $H^{(j)}:\bit^{\log n}\times \Sigma \rightarrow \bit$ for $j\in [n]$
  and defines $\ora$ as follows:
  $\ora$ takes $x\in \bit^n$ and $(\vv^{(1)},...,\vv^{(n)})\in C^n$ as input, and outputs $G(x)$ if 
 there is $j\in [n]$ such that
 $f_C^{H^{(j)}}(\vv^{(j)})=x$ and outputs $\bot$ otherwise.\footnote{Note that $x$ here plays the role of $r$ in \Cref{cor:YZLiu_separation_classical_query_amp}.} For simplicity we will denote by $\mathcal{H} = (H^{(1)}, \dots, H^{(n)})$.
\end{description} 
First, we show the easiness with quantum advice (\Cref{item:separation_poly_qpoly_dist_quantum_easiness}). 
Let $(\A',\{\ket{z'_\mathcal{H}}\}_\mathcal{H})$ be the tuple of an algorithm and family of quantum advice in \Cref{item:amp_quantum_easiness} of \Cref{cor:YZLiu_separation_classical_query_amp}. 
We construct $(\A,\{\ket{z_\ora}\}_{\ora})$ that satisfies \Cref{item:separation_poly_qpoly_dist_quantum_easiness} of \Cref{lem:separation_poly_qpoly_dist}.

In fact, we allow the advice $\ket{z_\ora}$ to be a mixed state and write it by $\rho_\ora$. Note that this does not weaken the statement since any mixed state can be considered as a distribution over pure states and thus there must exist a pure state advice $\ket{z_\ora}$ that is at least as good as $\rho_\ora$. 
The algorithm $\A$ and quantum advice $\rho_\ora$ is described as follows:
\begin{description}
    \item[$\rho_\ora$:] We describe a randomized procedure to set $\rho_\ora$ given an oracle $\ora$. 
    This should be understood as setting $\rho_\ora$ to be the mixed state corresponding to the output of the procedure.  
    Sample $(G,\mathcal{H})$ from the conditional distribution of $\dist_n$ conditioned on the given $\ora$. Note that then the joint distribution of $(G,\mathcal{H},\ora)$ is identical to the real one. Then $\rho_\ora$ is set to be $\rho_\mathcal{H}$.
    \item[$\A^\ora(\rho_\ora,x)$:] 
    It runs $\vv\la \A'(\rho_\ora,x)$, queries $(x,\vv)$ to $\ora$, and outputs whatever $\ora$ returns.  
\end{description}
Then \Cref{item:amp_quantum_easiness} of \Cref{cor:YZLiu_separation_classical_query_amp} implies
    \begin{align*}
        \Pr_{(G,\ora)\la \dist_n}[
    \forall x\in \bit^n~
    \Pr[\A^{\ora}(\rho_\ora,x)=G(x)]\ge 1-\negl(n)
    ]\ge  1-\negl(n).
    \end{align*}
Thus,  \Cref{item:separation_poly_qpoly_dist_quantum_easiness} of  \Cref{lem:separation_poly_qpoly_dist} holds.

Next, we show the hardness with classical advice (\Cref{item:separation_poly_qpoly_dist_classical_hardness}).
Suppose that there is $(\B,\{z_\ora\}_{\ora})$ that breaks it. Then we have 
    \begin{align*}
        \Pr_{\substack{(G,\ora)\la \dist_n \\x\la \bit^n}}[
    \B^{\ora}(z_\ora,x)=G(x) 
    ]
    > \frac{3}{5}
    \end{align*}
    for infinitely many $n\in \mathbb{N}$.  
Recall that $\ora$ returns $G(x)$ only if the query $(x,(\vv^{(1)},...,\vv^{(n)}))$ satisfies $f_C^{H^{(j)}}(\vv^{(j)})=x$ for some $j\in [n]$. 
 Thus, by a direct reduction to \Cref{lem:query_amplitude}, for a $\frac{1}{30}$  fraction of $x \in \bit^n$, the query list of $\B$ to a randomly chosen $\ora$ according to $\dist_n$ will contain a query of the form $(x,(\vv^{(1)},...,\vv^{(n)}))$ such that there is $j\in [n]$ such that $f_C^{H^{(j)}}(\vv^{(j)})=x$ with probability at least $\frac{1}{20}$. 

Also note that classical access to $\ora$ can be simulated by classical access to $G$ and $\mathcal{H} = (H^{(1)}, \dots, H^{(n)})$. 
Thus, the above directly gives an algorithm that violates the hardness with classical advice and classical access to $\mathcal{H}$ (\Cref{item:amp_classical_hardness} of \Cref{cor:YZLiu_separation_classical_query_amp}). To show this, we go through the following steps:
\begin{enumerate}
    \item For each fixed $G$, we define a pair $(\B'[G],\{z'_\mathcal{H}[G]\}_\mathcal{H})$ of an adversary and advice in which $G$ is hardwired.
    \item Show that $(\B'[G],\{z'_\mathcal{H}[G]\}_\mathcal{H})$  breaks \Cref{item:amp_classical_hardness} of \Cref{cor:YZLiu_separation_classical_query_amp} on average over the choice of $G$.
    \item Fix the ``best'' $G$ (w.r.t. random $\mathcal{H}$) to get a fixed pair of algorithm and advice that breaks  \Cref{item:amp_classical_hardness} of \Cref{cor:YZLiu_separation_classical_query_amp}.
\end{enumerate}
Specifically, it works as follows. 
For each $G$, we define $(\B'[G],\{z'_\mathcal{H}[G]\}_\mathcal{H})$ as follows:
\begin{description}
    \item[$z'_\mathcal{H}{[}G{]}$:] Construct $\ora$ from $(G,\mathcal{H})$. Set $z'_\mathcal{H}[G]:=z_\ora$.
    \item[${\B'[G]}^{\mathcal{H}}(z'_\mathcal{H}{[}G{]},x)$:] 
    It runs $\B^{\ora}(z'_\mathcal{H}[G],x)$ where $\B'$ simulates the oracle $\ora$ for $\B$ by using its own oracle $\mathcal{H}$ and the hardwired oracle $G$ and outputs a uniformly chosen query by $\B$. 
\end{description}
By the above argument, we have 
     \begin{align*}
        \Pr_{G,\mathcal{H},x}[
   \exists j\in [n]~\text{s.t.}~ f_C^{H^{(j)}}(\vv^{(j)})=x: (\vv^{(1)},...,\vv^{(n)})\la \B'[G]^\mathcal{H}(z'_\mathcal{H}[G],x)
    ]\ge \frac{1}{600Q}
    \end{align*}
    where $Q$ is the number of queries by $\B$. 
Thus, by taking $G=G^*$ that makes the above probability the largest, 
$(\B',\{z'_\mathcal{H}\}_\mathcal{H}):=(\B'[G^*],\{z'_\mathcal{H}[G^*]\}_\mathcal{H})$ breaks
\Cref{item:amp_classical_hardness} of \Cref{cor:YZLiu_separation_classical_query_amp}.\footnote{At first glance, this argument seems to allow $\B'$ to be a non-uniform machine that takes $G$ as advice. 
However, this is not needed since $\B'$ can find the best $G$ by itself by using its unbounded-time computational power.
} 
\end{proof}

Given \Cref{lem:separation_poly_qpoly_dist}, it is straightforward to prove  
a separation between $\BQP/\qpoly$ and $ \BQP/\poly$ relative to a classically-accessible classical oracle 
by the standard diagonalization argument. 
\begin{theorem}\label{thm:separation_poly_qpoly}
There is a classically-accessible classical oracle $\ora$ relative to which $\BQP/\poly\neq \BQP/\qpoly$.
\end{theorem} 
\begin{proof}
    Suppose that we generate $(G,\ora)\la \dist_n$ and define a language $\calL:=\bigsqcup_{n\in \mathbb{N}}G_n^{-1}(1)$ and an oracle $\ora$ that returns $\ora_{|x|}(x)$ on a query $x\in \bit^*$.  
    We claim that $\calL\in \BQP^{\ora}/\qpoly$ and $\calL\notin \BQP^{\ora}/\poly$ with probability $1$. 
    
    To see $\calL\in \BQP^{\ora}/\qpoly$ with probability $1$, \Cref{item:separation_poly_qpoly_dist_quantum_easiness} of \Cref{lem:separation_poly_qpoly_dist} implies that there is 
    a
    $\BQP$ machine $\A$ with polynomial-size quantum advice that decides $\calL$ on all $x$ of length $n$  with probability at least $1-\frac{1}{n^2}$ for all sufficiently large $n$. 
    Since $\sum_{n=1}^{\infty}\frac{1}{n^2}=\pi^2/6$ is finite, the Borel–Cantelli lemma implies that the probability that $\A$ fails on infinitely many $n$ is $0$.  
    In other words, the probability that there is $N$ such that  $\A$ succeeds in deciding $\calL$ on all $x$ such that $|x|\ge N$ is $1$.
    By augmenting $\A$ to decide $\calL$ by brute-force when the instance has length smaller than $N$, we can conclude that there is a  $\BQP$ machine with polynomial-size quantum advice that decides $\calL$ on all $x\in \bit^*$ with probability $1$ over the random choice of $(G,\ora)$ for $n\in \mathbb{N}$. 

    Next, we prove $\calL\notin \BQP^{\ora}/\poly$ with probability $1$.  
    For a $\BQP$ machine $\B$ that takes $\ell(n)$-bit classical advice for a polynomial $\ell$, we define $S_\B(n)$ to be the event over the choice of $(G,\ora)$ that there is a $\ell(n)$-bit classical advice $z_\ora$
    such that 
    $$
    \Pr[
    \forall x\in\bit^n~
    \B^{\ora}(z_\ora,x)=G(x)]\ge \frac{2}{3}.
    $$
   \Cref{item:separation_poly_qpoly_dist_classical_hardness}  of \Cref{lem:separation_poly_qpoly_dist} implies that there is an integer $N$ such that
    for any $\BQP$ machine $\B$ with classical access to $\ora$, 
    we have $\Pr_{G,\ora}[S_\B(n)]\le c$ for all $n\ge N$ where $c:=9/10$. 
     We now show that 
    $$
    \Pr_{G,\ora}[S_\B(1) \land S_\B(2) \land \dots]= 0$$
    \begin{itemize}
        \item We will consider a sequence of input lengths $n_1, n_2, \dots$ defined by 
            $n_1:= N$ and $n_i := T(n_{i-1}) + 1$, where $T(n)$ is the running time of $\B$ on input of length $n$. This means that when $\B$'s input length is $n_{i-1}$, it cannot touch the oracle on input length $\ge n_i$. This guarantees that $\Pr[S_\B(n_i) \mid S_\B(n_{j})] = \Pr[S_\B(n_i)]$ for all $i > j$.
        \item We can now show that the probability that $\B$ succeeds on all inputs is equal to $0$ over the choices of $G,\ora$.
        \begin{align*}
            &\Pr[S_\B(1) \land S_\B(2) \land \dots] \\
            &\leq \Pr\left[ \bigwedge_{i} S_\B(n_i)\right] \\
            &= \Pr[S_\B(n_1)] \cdot \Pr[S_\B(n_2) \mid S_\B(n_1)] \cdot \dots \\
            &\leq c \cdot c \cdot \dots \\
            &= 0
        \end{align*}
    \end{itemize}
    Since there are countably many QPT machines, 
    $$
    \Pr_{G,\ora}[\exists \B~S_\B(1) \land S_\B(2) \land \dots]=0.
    $$
This means that $\calL \not\in \BQP^\ora/\poly$ with probability $1$ over the choice of $(G,\ora)$.
\end{proof}

\paragraph{Stronger separations.}
We can easily extend our proof to show $\BQP/\qpoly\cap \NP \cap \coNP\not\subseteq \BQP/\poly$ relative to a classically-accessible classical oracle. This can be seen by observing that for any $x\in \bit^n$, we can use $(\vv^{(1)},...,\vv^{(n)})\in C^n$ such that $f_C^{H^{(j)}}(\vv^{(j)})=x$ for some $j\in [n]$ as a witness that certifies $G(x)$. This means that the language $\mathcal{L}$ in the proof of \Cref{thm:separation_poly_qpoly} is in  $ \NP \cap \coNP$.
Moreover, we can further strengthen the separation to show $\mathsf{YQP}\cap \NP \cap \coNP\not\subseteq \BQP/\poly$ relative to a classically-accessible classical oracle. Here, $\mathsf{YQP}$ is the class of problems that can be decided by a $\BQP$ machine with \emph{untrusted} quantum advice, which was introduced in \cite{Aaronson07} followed by a minor definitional modification in \cite{SICOMP:AarDru14}. This can be seen by 
observing that the $\BQP/\qpoly$ algorithm for deciding the language $\mathcal{L}$ works even if the advice is untrusted because its guess of $G(x)$ is correct whenever the oracle $\ora$ does not return $\bot$ on $(x,(\vv^{(1)},...,\vv^{(n)}))$ where $\vv^{(j)}$ are candidate solutions to the YZ problem w.r.t. $H^{(j)}$ and $x$ generated by using the given (potentially ill-formed) quantum advice. 
\section{\texorpdfstring{$\QMA$ vs $\QCMA$ under Classically-Accessible Oracle}{QMA vs QCMA under Classically-Accessible Oracle}}\label{sec:QMA_QCMA_classical}
In this section, we demonstrate a $\QMA$ and $\QCMA$ separation relative to a classically-accessible classical oracle. 

   \paragraph{Notation.} Let $C$ be the code in \Cref{thm:YZLiu_separation}. (Remark that $C$ and $\Sigma$ are actually indexed by $n$, but we omit it for notational simplicity. See \Cref{rem:subscript_n}.)
    For $n\in \mathbb{N}$, 
    functions $G:\bit^n\ra \bit^n$ and $H:[n]\times \Sigma \ra \bit$, and a subset $S\subseteq \bit^n$, we define the following oracle:
\begin{description}
    \item[$\ora_{n}{[}G,H,S{]}$:] 
    It takes $t\in \bit^n$ and $\vv \in C$ as input and outputs $1$ if $f_{C}^{H}(\vv)=G(t)$ and $t\notin S$ and otherwise outputs $0$. 
\end{description}
The following is the key technical lemma for the separation between $\QMA$ and $\QCMA$ relative to a classically-accessible oracle. 
\begin{lemma}\label{lem:core_separation_QMA_QCMA_classical_access} 
The following hold:
\begin{enumerate} 
\item \label{item:dist_with_quantum} {\bf (Distinguishability with Quantum Witness)} There is a QPT algorithm $\A$ that makes polynomially many classical queries and a polynomial $\ell$  such that 
the following hold:
\begin{enumerate}
    \item \label{item:dist_with_quantum_completeness}
    There is an $\ell(n)$-qubit state $\ket{w}$ such that 
    $$\Pr_{G,H}[\A^{G,\ora_{n}[G,H,\emptyset]}(1^n,\ket{w})=1]\ge 1-\negl(n)$$ 
    where $G\la \Func(\bit^n,\bit^n)$ and $H\la \Func([n]\times \Sigma, \bit)$.
    \item \label{item:dist_with_quantum_soundness}
    For any
    $G$, $H$, 
    $S\subseteq \bit^n$, 
    and $\ell(n)$-qubit state $\ket{w}$, 
    $$\Pr[\A^{G,\ora_{n}[G,H,S]}(1^n,\ket{w})=1]\le 1-\frac{|S|}{2^n}$$
    for all $n\in \mathbb{N}$.
\end{enumerate}
\item \label{item:ind_with_classical}  {\bf (Indistinguishability with Classical Witness)}
For any unbounded-time algorithm $\B$ that makes polynomially many classical queries and polynomial $s$, 
there is a negligible function $\mu$ 
such that for any
family of $s(n)$-bit classical witness $\{w_{G,H}\}_{G, H}$,
there is a family $\{S_{G, H}\}_{G, H}$ with $|S_{G, H}| \geq 2^n\cdot 2/3$ and   
$$\left|\Pr_{G,H}[\B^{G,\ora_{n}[G,H,\emptyset]}(1^n,w_{G,H})=1]-\Pr_{G,H}[\B^{G,\ora_{n}[G,H,S_{G, H}]}(1^n,w_{G,H})=1]\right|\le \mu(n).$$
\end{enumerate}
\end{lemma}
\begin{proof}
We start with the distinguishability with quantum witness. 

\noindent\textbf{Distinguishability with quantum witness.} We use the quantum advice of~\Cref{thm:YZLiu_separation} as a witness. The algorithm $\A$ proceeds as follows: Sample a random $t \in \bit^n$ and compute $r = G(t)$. Note that $r$ is also uniformly random over $\bit^n$.
Since we can generate $\vv\in C$ such that $H(\vv)=r$ using the algorithm in~\Cref{thm:YZLiu_separation} with probability~$1-\negl(n)$ over random $H,r$, $\A$ queries $\ora_n[G, H,S]$ with the generated $(t,\vv)$, and output the query result. If $S=\emptyset$, the oracle should return 1 with probability $1-\negl(n)$, else returns $0$ if the random $t \in S$, which happens with probability $\frac{|S|}{2^n}$.

\smallskip
\noindent\textbf{Indistinguishability with classical witness.} 
In the following proof, unless specified otherwise, we assume $G, H$ are uniformly sampled when using the notation $\Pr_{G, H}[X]$. 

The advantage for an
unbounded-time algorithm $\B$ with classical oracle queries to distinguish $(G,\ora_{n}[G, H, \emptyset])$ from
$(G, \ora_{n}[G, H, S_{G, H}])$ is at most the probability the queries of $\B$ include an input $(t, \vv)$ on which the two oracles differ. These inputs are precisely the ones that satisfy $t\in S_{G, H}$ and $f_C^H(\vv) = G(t)$. Thus we define $\tilde{\B}$ to be the adversary that is given oracle access to $(G,\ora_{n}[G, H, \emptyset])$, the same input as $\B$, and the set $S_{G, H}$ and outputs the first query $(t, \vv)$ of $\B$ on which the two oracles differ, i.e. $t \in S_{G, H}$ and $f_C^H(\vv) = G(t)$. (If there is no such query, $\tilde{\B}$ aborts.) Then we have the following inequality: 
  \begin{align*}
      \left|\Pr_{G, H}[\B^{G, \ora_{n}[G, H,\emptyset]}(1^n, w_{G, H})=1]-\Pr_{G, H}[\B^{G, \ora_{n}[G, H,S_{G, H}]}(1^n, w_{G, H})=1]\right|\\
      \leq \Pr_{G, H}\left[t \in S_{G, H} \land f_C^H(\vv) = G(t) : (t, \vv) \leftarrow \tilde{\B}^{G, \ora_{n}[G, H, \emptyset]}(1^n, w_{G, H}, S_{G, H})\right]
  \end{align*}
Note that one can simulate the oracle $(G, \ora_n[G, H, \emptyset])$ using only access to $(G, H)$. Thus it suffices to prove the following: 

\begin{claim}\label{cla:prob_find_difference}
    For any unbounded-time $\tilde{\B}$ that makes polynomially many classical queries and polynomial $s$, there is a negligible function $\mu$ such that for any
    family of $s(n)$-bit classical advice $\{w_{G, H}\}_{G, H}$, there is a family $\{S_{G, H}\}_{G, H}$ of subsets of $\{0, 1\}^n$ such that $|S_{G, H}| \geq 2^n\cdot 2/3$ and
    $$\Pr_{G, H}[t \in S_{G, H} \land f_C^H(\vv) = G(t) : (t, \vv) \la \tilde{\B}^{G, H}(1^n, w_{G, H}, S_{G, H})] \leq \mu(n).$$ 
\end{claim}

We prove \Cref{cla:prob_find_difference} by reducing it to a similar statement with $(P', 1-\delta)$-dense distributions by using \Cref{cl:CDGSalmostbitfixing}.\footnote{We cannot simply apply \Cref{thm:classical_bf_to_nonuniform} here because the probability considered in \Cref{cla:prob_find_difference} is not captured by security of a game in the ROM as defined in \Cref{def:game_ROM}.} 

    Let $Q$ be the number of queries by $\tilde{\B}$. 
    Applying \Cref{cl:CDGSalmostbitfixing} and \Cref{cl:CDGS-length-of-fixed}, for all $\gamma, \delta > 0$,
    there is a family $\{\calD_{G,H}\}_{G,H}$ of convex combinations of finitely many $(P', 1-\delta)$-dense distributions
    $\calD_{G,H}$ such that for any family $\{S_{G, H}\}_{G, H}$ of subsets of $\bit^n$, 
    \begin{align}
    \begin{split}
        &\Big|\Pr_{G, H}[t \in S_{G, H} \land f_C^H(\vv) = G(t) : (t, \vv) \la \tilde{\B}^{G, H}(1^n, w_{G, H}, S_{G, H})]  \\
        &- \Pr_{\substack{G,H\\(G',H')\la \calD_{G,H}}}[t \in S_{G', H'} \land f_C^{H'}(\vv) = G'(t) : (t, \vv) \la \tilde{\B}^{G', H'}(1^n,w_{G',H'}, S_{G', H'})]\Big|
        \quad \leq 2\gamma,
    \end{split} \label{eq:AI_to_BF}
    \end{align}
    where $P' = \frac{s + 2\log 1/\gamma}{\delta\cdot n}$. 
    Here we used \Cref{cl:CDGSalmostbitfixing} which states that $\calD_{G, H}$ is $\gamma$-close to $G, H$ for $P' = \frac{s_{G, H} + \log1/\gamma}{\delta\cdot n}$, with $s_{G, H}$ being the min-entropy deficiency of $G, H$ conditioned on $w_{G, H}$. Then we use \Cref{cl:CDGS-length-of-fixed} and ``truncate'' $s_{G, H}$ to a maximum of $s + \log 1/\gamma$ by only increasing the distance between the two distributions by an additional $\gamma$ additive factor.



Let us now construct a subset $S_{G', H'}$ that satisfies the statement of our claim. We will begin by finding a randomized construction of such a subset $S^*$ that will depend on all $(G, H, G', H')$, and we will derandomize it to $S_{G', H'}$. For each $(G,H)$, since $\calD_{G,H}$ is a convex combination of $(P', 1-\delta)$-dense distributions, it can be written as  
$$
\calD_{G,H}=\sum_{i}p_{G,H,i} \calD_{G,H,i}
$$
where $0<p_{G,H,i} \le 1$, 
$\sum_{i}p_{G,H,i}=1$ and each $\calD_{G,H,i}$ is a $(P', 1-\delta)$-dense distribution. 
Let $\mathcal{I}_{G,H}$ be a distribution that samples $i$ with probability $p_{G,H,i}$. 

\paragraph{Randomized $S^*$.} The randomized subset $S^*$ is defined using \Cref{alg:sampling}.

\begin{algorithm}
    \caption{Sampling $G', H', S^*$ from $\calD^1_{G, H}$}
    \label{alg:sampling}
    \begin{algorithmic}[1]
        \STATE Sample $i$ according to $\mathcal{I}_{G, H}$.
        \STATE Sample $G', H'$ from $\calD_{G, H, i}$.
        \STATE Define $S^* := \{\text{set of indices not fixed by}~\calD_{G, H, i}\}$.
    \end{algorithmic}
\end{algorithm}

We proceed with proving \Cref{cla:prob_find_difference} 
with this randomized $S^*$.
We will use the notation $(G', H', S^*) \la \calD^1_{G, H}$ to mean that we sample $(G', H', S^*)$ according to \Cref{alg:sampling}. 
Since each $\calD_{G,H,i}$ is a $(P', 1-\delta)$-dense distribution, we have that $|S^*| \geq 2^n - P'$ always.

      Define $\calD^{1, *}_{G,H}$ to be the distribution that samples $(G',H')$ by first sampling $(G',H', S^*)\leftarrow \calD^1_{G,H}$ and then replacing $G'(t)$ on all $t\in S^*$ by resampling according to the $(1-\delta)$-dense distribution indicated by $\calD_{G, H, i}$. By definition, since we replace the non-fixed values of $G'$ with a new sample from the same $(1 - \delta)$-dense distribution, the two distributions of $G'(t)$ for 
      $(G',H', S^*)\leftarrow \calD^1_{G,H}$ and $(G',H', S^*)\leftarrow \calD^{1, *}_{G,H}$ 
      are the same. Thus, we have 

    \begin{align}
    \begin{split}
        &\Pr_{\substack{G,H\\(G',H', S^*)\la\calD^1_{G,H}}}\left[t \in S^* \land f_C^{H'}(\vv) = G'(t) \colon (t, \vv )\la \tilde{\B}^{G', H'}(1^n,w_{G',H'}, S^*)\right]\\
        &=~\Pr_{\substack{G,H\\(G',H', S^*) \la\calD^{1, *}_{G,H}}}
        \left[t \in S^* \land f_C^{H'}(\vv) = G'(t) \colon (t, \vv )\la \tilde{\B}^{G', H'}(1^n,w_{G',H'}, S^*)\right].
    \end{split} \label{eq:replace_S}
    \end{align}
    
    Now we prove that 
    \begin{align}
        \Pr_{\substack{G,H\\(G',H', S^*) \la\calD^{1, *}_{G,H}}}[t \in S^* \land f_C^{H'}(\vv) = G'(t) \colon (t, \vv) \la \tilde{\B}^{G', H'}(1^n,w_{G',H'}, S^*)]\leq \negl(n). \label{eq:BF_negl}
    \end{align}
    Note that
     \begin{align*}
        &\Pr_{\substack{G,H\\(G',H', S^*) \la\calD^{1,*}_{G,H}}}[t \in S^* \land f_C^{H'}(\vv) = G'(t) \colon (t, \vv) \la \tilde{\B}^{G',H'}(1^n, w_{G',H'}, S^*)]\\
        &=\sum_{i\in [Q_{G'}]}\left(\Pr_{\substack{G,H\\(G',H', S^*)\la\calD^{1, *}_{G,H}}}\left[t=t_i \land t \in S^* \land f_C^{H'}(\vv) = G'(t) \colon (t, \vv) \la \tilde{\B}^{G',H'}(1^n,w_{G',H'}, S^*)\right]\right)\\&
        +\Pr_{\substack{G,H\\(G',H', S^*)\la\calD^{1, *}_{G,H}}}[t\notin\{t_1,...,t_{Q_{G'}}\} \land t \in S^* \land f_C^{H'}(\vv) = G'(t) \colon (t, \vv) \la \tilde{\B}^{G',H'}(1^n, w_{G',H'},S^*)].
    \end{align*}
     where $Q_{G'}$ is the number of queries to $G'$ and $t_i$ is the $i$-th query to $G'$. 
    For each $i$, we will set the parameters $\gamma = 2^{-n}$ and $\delta = 2^{-n^{c''}}$ (for a sufficiently small constant $c''$ that will be specified below) to bound the first term as follows:
    $$\Pr_{\substack{G,H\\(G',H', S^*)\la\calD^{1, *}_{G,H}}}\left[t=t_i \land t \in S^* \land f_C^{H'}(\vv) = G'(t) \colon (t, \vv) \la \tilde{\B}^{G',H'}(1^n,w_{G',H'}, S^*)\right]\le \negl(n).$$


\takashi{I believe the above statement assumes that $\delta$ and $\gamma$ are not too small. In this case, I believe the choice of $\delta$ and $\gamma$ should be specified before stating the above inequality.}
\angelos{added some justification for our choice of parameters}
\takashi{Looks good}

    This is because of the following argument: We can decompose the above probability as a weighted sum of probabilities over the $(P', 1-\delta)$-dense distributions $\calD_{G, H, i}$. Then, from \Cref{cl:dense-to-bit-fixing}, we can replace each $\calD_{G, H, i}$ with its corresponding $P'$-bit-fixing distribution $\calD'_{G, H, i}$, while only incurring an error of at most $Q\delta\cdot n$. Recall that $Q = \poly(n)$ is the number of queries that $\tilde{\B}$ makes to the oracle, and we have upper-bounded $\log |Y|$ by $n$, since $G'$ maps inputs to $\{0, 1\}^n$. Setting $\delta = 2^{n^{-c''}}$ for a sufficiently small $c''$ \footnote{In particular, we want $c''$ to be smaller than the constant $c$ that appears in the list-recoverability of the code in \Cref{thm:YZLiu_separation} so we can apply the same proof as \Cref{thm:YZLiu_separation_classical_query} to show the hardness of YZ with $P'$-bit-fixing oracles.}, gives that by substituting the corresponding bit-fixing distributions the advantage of the adversary only increases negligibly. 
    
     Now it suffices to show the hardness of the YZ problem when $G', H'$ are drawn from a $P'$-bit-fixing distribution. Recall that $P' = \frac{s + 2\log 1/\gamma}{\delta\cdot n}$, and thus $P' = \Theta\left(2^{n^{c''}}\poly(n)\right)$ for our choice of parameters. The hardness of the YZ problem under $P'$-bit-fixing oracles can be shown by repeating exactly the same proof as that of \Cref{thm:YZLiu_separation_classical_query} noting that $G'(t_i)$ is uniformly random for $\tilde{\B}$ before querying it when $t \in S^*$ and thus we can embed a fresh problem instance $r \la \bit^n$ into $G'(t_i)$. For the same proof to go through, we also need the fact that $P' =\Theta\left(2^{n^{c''}}\poly(n)\right) < 2^{n^c}$, where $c<1$ is the positive constant such that the code $C$ satisfies $(\zeta, 2^{n^c}, L)$-list-recoverability. \takashi{It may be better to mention that we are using $P'<2^{n^c}$ here.}
     \angelos{updated the justification for $P'$}

     \angelos{end of the proof above}

    The second term
        $$\Pr_{\substack{G,H\\(G',H', S^*)\la\calD^{1, *}_{G,H}}}[t\notin\{t_1,...,t_{Q_{G'}}\} \land t \in S^* \land f_C^{H'}(\vv) = G'(t) \colon (t, \vv) \la \tilde{\B}^{G',H'}(1^n, w_{G',H'},S^*)]$$
    is bounded by $2^{-n} + \negl(n)$. We again use \Cref{cl:dense-to-bit-fixing} with the same choice of parameters $(\gamma, \delta)$ to replace $(G', H')$ with their bit-fixing counterparts, while incurring negligible error. Now $G'(t)$ is uniformly random for $\tilde{\B}$ (since it was never queried) and the probability that $f_C^{H'}(\vv)$ is equal to a uniformly random value is $2^{-n}$. 
    Combining the above and $Q_{G'}=\poly(n)$, we get \Cref{eq:BF_negl}.

    \paragraph{Derandomizing $S^*$.} Our goal is to show that there exists a family of subsets $\{S_{G, H}\}_{G, H}$ such that the statement holds. To do this, we will take our randomized construction of $S^*$ and derandomize it to obtain a fixed subset $S_{G', H}$ for any $G', H'$. To this end, we begin with sampling $G', H', S^*$ according to $\calD^2_{G, H}$, whose description is given in \Cref{alg:equivalent-sampling}.

    \begin{algorithm}
    \caption{Sampling $G', H', S^*$ from $\calD^2_{G, H}$}
    \label{alg:equivalent-sampling}
    \begin{algorithmic}[1]
        \STATE Sample $i$ according to $\mathcal{I}_{G, H}$.
        \STATE Sample $G', H'$ from $\calD_{G, H, i}$.
        \STATE ($\star$) Sample $i'$ from the conditional distribution of $i$ conditioned on $(G, H, G', H')$. Formally, sample $i'$ with probability
        $$\frac{p_{G, H, i'}\cdot\Pr[(G', H')\la \calD_{G, H, i'}]}{\sum_i p_{G, H, i}\cdot\Pr[(G', H')\la \calD_{G, H, i}]}$$
        \STATE ($\star$) Define $S^* := \{\text{set of indices not fixed by}~\calD_{G, H, i'}\}$.
    \end{algorithmic}
\end{algorithm}

\angelos{I don't know if the notation I have here is very friendly with $\mathcal{X}$ etc...}

Note that the distribution of $(G', H')$ from steps 1 and 2 is $\calD_{G, H}$. 
We will also denote with $\mathcal{S}_{G,H,G',H'}$ the distribution of $S^*$ from steps 3 and 4 conditioned on $(G, H, G', H')$. It is not hard to see that the resulting distributions $\calD^1_{G, H}$ and $\calD^2_{G, H}$ are equal, and thus \Cref{eq:BF_negl} gives:
\begin{align*}
    \Pr_{\substack{G,H\\(G',H', S^*) \la\calD^{2}_{G,H}}}[t \in S^* \land f_C^{H'}(\vv) = G'(t) \colon (t, \vv) \la \tilde{\B}^{G', H'}(1^n,w_{G',H'}, S^*)]\leq \negl(n).
\end{align*}
Writing $\calD^2_{G, H}$ in terms of $\calD_{G, H}$ and $\mathcal{S}_{G, H, G', H'}$ gives:
\begin{align}
    \Pr_{\substack{G,H\\(G',H') \la\calD_{G,H} \\ S^* \la \mathcal{S}_{G, H, G', H'}}}[t \in S^* \land f_C^{H'}(\vv) = G'(t) \colon (t, \vv) \la \tilde{\B}^{G', H'}(1^n,w_{G',H'}, S^*)]\leq \negl(n).
\end{align}

What we now have is that for a set $S^*$ sampled according to the distribution $\mathcal{S}_{G, H, G', H'}$, the value of the LHS is negligible. Thus, there must exist some fixed subset $S_{G, H, G', H'}$ for every $(G, H, G', H')$ such that
\begin{align}
    \Pr_{\substack{G,H\\(G',H') \la\calD_{G,H}}}[t \in S_{G, H, G', H'} \land f_C^{H'}(\vv) = G'(t) \colon (t, \vv) \la \tilde{\B}^{G', H'}(1^n,w_{G',H'}, S_{G, H, G', H'})]\leq \negl(n). \label{eq:S-depends-on-GHGH}
\end{align}
    
    Additionally, the above probability never uses the original $G, H$ oracles. Hence we can apply a trivial averaging argument to conclude that there must exist some fixed $S_{G', H'}$ that only depends on $G', H'$ such that
    \begin{align}
        \Pr_{\substack{G,H\\(G',H') \la\calD_{G,H}}}[t \in S_{G', H'} \land f_C^{H'}(\vv) = G'(t) \colon (t, \vv) \la \tilde{\B}^{G', H'}(1^n,w_{G',H'}, S_{G', H'})]\leq \negl(n). \label{eq:S-depends-on-GH}
    \end{align}
Now we combine 
    \Cref{eq:AI_to_BF,eq:S-depends-on-GH} 
     where we set $\delta = 2^{-n^{c''}}$, $s=\poly(n)$, $Q=\poly(n)$, and $\gamma=2^{-n}$, which implies that $P' = \Theta(2^{n^{c''}} \poly(n))$, 
     and we conclude that
    $$\Pr_{G, H}[t \in S_{G, H} \land f_C^H(\vv) = G(t) : (t, \vv) \la \tilde{\B}^{G, H}(1^n, w_{G, H}, S_{G, H})] \leq \negl(n).$$

    Finally, recall that $|S_{G, H}| \geq 2^n - P' \geq 2^n\cdot 2/3$ and thus our condition on the sizes of the subsets is also satisfied.
This completes the proof of \Cref{cla:prob_find_difference}, which in turn completes the proof of \Cref{lem:core_separation_QMA_QCMA_classical_access}.

\end{proof}
Next, we give a corollary of \Cref{lem:core_separation_QMA_QCMA_classical_access}, which is useful for showing the separation between $\QMA$ and $\QCMA$ relative to a classically-accessible oracle.
\begin{corollary}\label{cor:diagonalization}
Let $\A$ and $\ell$ be as in \Cref{lem:core_separation_QMA_QCMA_classical_access}.
For any unbounded-time algorithm $\B$ that makes polynomially many classical queries and polynomial $s$, there is an integer $N_{\B,s}$ 
such that either of the following holds for all $n\ge N_{\B,s}$:
\begin{enumerate} 
    \item \label{item:diagonalization_case_one}
    There are $G\in \Func(\bit^n,\bit^n)$ and $H\in \Func([n]\times \Sigma, \bit)$ such that 
    \begin{enumerate}
    \item \label{item:diagonalization_case_one_a}
    There is an $\ell(n)$-qubit state $\ket{w}$ such that $\Pr[\A^{G,\ora_{n}[G,H,\emptyset]}(1^n,\ket{w})=1]\ge 2/3$.
    \item \label{item:diagonalization_case_one_b}
    $\Pr[\B^{G,\ora_{n}[G,H,\emptyset]}(1^n,w)=1]< 2/3$ for any $w\in \bit^{s(n)}$. 
    \end{enumerate}
        \item \label{item:diagonalization_case_two} 
        There are $G\in \Func(\bit^n,\bit^n)$, $H\in \Func([n]\times \Sigma, \bit)$,  and $S\subseteq \bit^n$ such that  
    \begin{enumerate}
    \item \label{item:diagonalization_case_two_a}
     $\Pr[\A^{G,\ora_{n}[G,H,S]}(1^n,\ket{w})=1]\le 1/3$ for any $\ell(n)$-qubit state $\ket{w}$.
    \item \label{item:diagonalization_case_two_b}
    There is $w\in \bit^{s(n)}$ such that $\Pr[\B^{G,\ora_{n}[G,H,S]}(1^n,w)=1]>  1/3$. 
    \end{enumerate}
\end{enumerate}
\end{corollary}
\begin{proof}
  By  \Cref{item:dist_with_quantum_completeness} of \Cref{lem:core_separation_QMA_QCMA_classical_access} and a standard averaging argument, 
  for $1-\negl(n)$-fraction of $(G,H)$,  
    there is an $\ell(n)$-qubit state $\ket{w}$ such that 
    $$\Pr[\A^{G,\ora_{n}[G,H,\emptyset]}(1^n,\ket{w})=1]\ge \frac{2}{3}$$
for sufficiently large $n$. Let $\good_n$ be the set of all such $(G,H)$. 
Suppose that there is $(G,H)\in \good_n$ such that $\Pr[\B^{G,\ora_{n}[G,H,\emptyset]}(1^n,w)=1]< 2/3$ for any $w\in \bit^{s(n)}$. Then it implies that \Cref{item:diagonalization_case_one} of \Cref{cor:diagonalization} is satisfied. 
Thus, it suffices to prove that \Cref{item:diagonalization_case_two} of \Cref{cor:diagonalization} is satisfied assuming that 
for all $(G,H)\in \good_n$, there is $w\in \bit^{s(n)}$ such that $\Pr[\B^{G,\ora_{n}[G,H,\emptyset]}(1^n,w)=1]\ge 2/3$. We prove it below.

 Since $\good_n$ consists of $1-\negl(n)$-fraction of $(G,H)$, 
a similar inequality to \Cref{item:ind_with_classical} of \Cref{lem:core_separation_QMA_QCMA_classical_access} holds even if we sample $(G,H)$ from $\good_n$, i.e., 
there is a negligible function $\mu'$ 
such that for any
family of $s(n)$-bit classical advice $\{w_{G,H}\}_{G,H}$, 
there is a subset $S\subseteq \bit^n$ such that $|S|\ge 2^{n}\cdot 2/3$ and  
$$\left|\Pr_{G,H}[\B^{G,\ora_{n}[G,H,\emptyset]}(1^n,w_{G,H})=1]-\Pr_{G,H}[\B^{G,\ora_{n}[G,H,S]}(1^n,w_{G,H})=1]\right|\le \mu'(n)$$
 where $(G,H)\la \good_n$. In particular, there is $(G,H)\in \good_n$ such that for any  $s(n)$-bit classical advice $w$, there is a subset $S\subseteq \bit^n$ such that $|S|\ge 2^{n}\cdot 2/3$ and
 $$\left|\Pr_{G,H}[\B^{G,\ora_{n}[G,H,\emptyset]}(1^n,w)=1]-\Pr_{G,H}[\B^{G,\ora_{n}[G,H,S]}(1^n,w)=1]\right|< \frac{1}{3}$$
 for sufficiently large $n$.
 We fix such $(G,H)$. Since $(G,H)\in \good_n$, by our assumption, there is $w\in \bit^{s(n)}$ such that $\Pr[\B^{G,\ora_{n}[G,H,\emptyset]}(1^n,w)=1]\ge 2/3$.
 Combined with the above inequality,  there are
 $w\in \bit^{s(n)}$ and
 a subset $S\subseteq \bit^n$ such that $|S|\ge 2^{n}\cdot 2/3$ and
 $\Pr[\B^{G,\ora_{n}[G,H,S]}(1^n,w)=1]> 1/3$. This means that \Cref{item:diagonalization_case_two_b} of \Cref{cor:diagonalization} is satisfied. 
 Moreover, since $|S|\ge 2^{n}\cdot 2/3$, \Cref{item:dist_with_quantum_soundness} implies \Cref{item:diagonalization_case_two_a} of \Cref{cor:diagonalization}. Thus, \Cref{item:diagonalization_case_two} of  \Cref{cor:diagonalization} is satisfied. This completes the proof of  \Cref{cor:diagonalization}.
\end{proof}

Given \Cref{cor:diagonalization}, it is straightforward to prove the separation of $\QMA$ and $\QCMA$ relative to a classically-accessible oracle using the standard diagonalization argument.  
\begin{theorem}\label{thm:separation_QMA_QCMA_classical_access}
There is a classically-accessible classical oracle relative to which $\QMA \neq \QCMA$.
\end{theorem} 
\begin{proof}[Proof of \Cref{thm:separation_QMA_QCMA_classical_access}]
We enumerate all tuples $(\B_1,s_1),(\B_2,s_2),...$ where $\B_j$ is a QPT machine that makes polynomially many classical queries and $s_j$ is a polynomial for $j\in \mathbb{N}$. 
Let $T_j$ be a polynomial such that $\B_j$ runs in time $T_j(n)$ when it takes $1^n$ and $y\in \bit^{s_j(n)}$ as input.  
For  
a sequence $\ora=(\ora_1,\ora_2,...)$ of oracles $\ora_n:\bit^n\times C_n \rightarrow \bit$, 
$n^*\in \mathbb{N}$, 
$G:\bit^{n^*}\rightarrow \bit^{n^*}$, 
$H:[n^*]\times \Sigma^{n^*}\ra \bit$, and $S\subseteq \bit^{n^*}$, 
let $\tilde{\ora}[G,H,S]$ be the same as $\ora$ except that the $n^*$-th oracle is replaced with $\ora_{n^*}[G,H,S]$. That is,
we define $\tilde{\ora}[G,H,S]:=(\tilde{\ora}_1[G,H,S],\tilde{\ora}_2[G,H,S],...)$ where  $\tilde{\ora}_{n^*}[G,H,S]:=\ora_{n^*}[G,H,S]$ and $\tilde{\ora}_{n}[G,H,S]:=\ora_{n}$ for all $n\neq n^*$.
We define sequences of oracles $G^{(0)},G^{(1)},...$ and $\ora^{(0)},\ora^{(1)},...$ and a sequence of bits $\flag_1,\flag_2,...$ by the following procedure, where 
for each $i$, 
$G^{(i)}$ and
$\ora^{(i)}$ themselves are also sequences of oracles 
$G^{(i)}_1,G^{(i)}_2,...$ and
$\ora^{(i)}_1,\ora^{(i)}_2,...$, respectively, such that 
$G^{(i)}_n:\bit^n\ra \bit^n$
and 
$\ora^{(i)}_n:\bit^n\times C \rightarrow \bit$.
\begin{enumerate}
    \item 
    Take 
    $(G_1^{(0)},G_2^{(0)},...)$ and
    $(H_1^{(0)},H_2^{(0)},...)$ in such a way that 
  there is a $\ell(n)$-qubit state $\ket{w}$ such that $$\Pr[\A^{G_n,\ora_{n}[G_n^{(0)},H_n^{(0)},\emptyset]}(1^n,\ket{w})=1]\ge 2/3$$ for all $n\in \mathbb{N}$ 
    where $\A$ and $\ell$ are as in 
    \Cref{item:dist_with_quantum} of
    \Cref{lem:core_separation_QMA_QCMA_classical_access}. Note that such $(G_1^{(0)},G_2^{(0)},...)$ and
    $(H_1^{(0)},H_2^{(0)},...)$ exist by \Cref{item:dist_with_quantum} of \Cref{lem:core_separation_QMA_QCMA_classical_access}.
    Let $\ora^{(0)}:=(\ora^{(0)}_1,\ora^{(0)}_2,...)$ where $\ora^{(0)}_n:=\ora_n[G_n^{(0)},H_n^{(0)},\emptyset]$ for $n\in \mathbb{N}$.
    Set $n_0:=1$ and
     $\flag_n:=1$ for all $n\in \mathbb{N}$. 
    \item For $i=1,2,...,$ 
    \begin{enumerate}
        \item 
        Let $n_i:=\max\{T_{i-1}(n_{i-1})+1,N_{\B_i,s_i}\}$ where $N_{\B_i,s_i}$ is as in \Cref{cor:diagonalization} where we define $T_0(n_0):=1$ for convenience.
        \item 
        Do either of the following:
        \begin{enumerate}
        \item 
        If there are $G_{n_i}^{(i)}$ and $H_{n_i}^{(i)}$ such that 
        \begin{itemize}
        \item
    there is an $\ell(n_i)$-qubit state $\ket{w}$ such that $\Pr[\A^{G^{(i)},\tilde{\ora}^{(i-1)}[G_{n_i}^{(i)},H_{n_i}^{(i)},\emptyset]}(1^{n_i},\ket{w})=1]\ge 2/3$, and
    \item $\Pr[\B_i^{G^{(i)},\tilde{\ora}^{(i-1)}[G_{n_i}^{(i)}H_{n_i}^{(i)},\emptyset]}(1^{n_i},w)=1]< 2/3$ for any $w\in \bit^{s_i(n_i)}$,
    \end{itemize}
     where $G_{n}^{(i)}:=G_{n}^{(i-1)}$ for all $n\neq n_i$ (which define $G^{(i)}:=(G_1^{(i)},G_2^{(i)},...)$), 
       then set            $\ora^{(i)}:=\tilde{\ora}^{(i-1)}[G_{n_i}^{(i)},H_{n_i}^{(i)},\emptyset]$. 
    \item Otherwise, by \Cref{cor:diagonalization},
    there are $G_{n_i}^{(i)}$, $H_{n_i}^{(i)}$, and  $S_{n_i}^{(i)}$  such that 
    \begin{itemize}
\item $\Pr[\A^{G^{(i)},\tilde{\ora}^{(i-1)}[G_{n_i}^{(i)},H_{n_i}^{(i)},S_{n_i}^{(i)}]}(1^{n_i},\ket{w})=1]\le 1/3$ for any $\ell(n_i)$-qubit state $\ket{w}$.
    \item There is $w\in \bit^{s_i(n_i)}$ such that $\Pr[\B_i^{G^{(i)},\tilde{\ora}^{(i-1)}[G_{n_i}^{(i)},H_{n_i}^{(i)},S_{n_i}^{(i)}]}(1^{n_i},w)=1]>  1/3$. 
    \end{itemize}
     where $G_{n}^{(i)}:=G_{n}^{(i-1)}$ for all $n\neq n_i$ (which define $G^{(i)}:=(G_1^{(i)},G_2^{(i)},...)$). 
     Set 
     $\ora^{(i)}:=\tilde{\ora}^{(i-1)}[G_{n_i}^{(i)},H_{n_i}^{(i)},S_{n_i}^{(i)}]$ and overwrite $\flag_{n_i}:=0$.
        \end{enumerate}
    \end{enumerate}
\end{enumerate}
Let 
$G$ and
$\ora$ be oracles that are consistent to 
$G^{(i)}$ and
$\ora^{(i)}$, respectively, on all $n\in [n_{i+1}-1]$ for all $i\in \{0\}\cup \mathbb{N}$.
(That is,  for any $i\in \{0\}\cup \mathbb{N}$, $n\in [n_{i+1}-1]$, $t\in \bit^n$, and $\vv\in C$,  
$G(t)=G_n^{(i)}(t)$ and $\ora(t,\vv)=\ora_n^{(i)}(t,\vv)$.)
They are well-defined since we have $G_n^{(i+1)}=G_n^{(i)}$ and $\ora_n^{(i+1)}=\ora_n^{(i)}$ for all $n\leq  n_{i+1}-1$ by the definitions of $G^{(i)}$ and
$\ora^{(i)}$.
Let 
$\mathcal{L}$ be a unary language defined as
$\mathcal{L}:=\{1^n: \flag_n=1\}$. Then by the definitions of $G$ and $\ora$, $\A$ is a valid $\QMA$ verification algorithm for $\mathcal{L}$ i.e., $\mathcal{L}\in \QMA^{G,\ora}$.\footnote{Strictly speaking, $\A$ may not work as a $\QMA$ verifier for small $n$ on which $\Pr[\A^{G_n,\ora_{n}[G_n^{(0)},H_n^{(0)},\emptyset]}(1^n,\ket{w})=1]\ge 2/3$ does not hold. 
However, since there are only finitely many such $n$, we can augment $\A$ to work on all $n\in \mathbb{N}$ by hardwiring the correct outputs on all such $n$.} 
Moreover, for any QPT machine $\B_i$ with classical witness length $s_i$, it fails to be a valid $\QCMA$ verification algorithm for $\mathcal{L}$ on input length $n_i$. Thus, we have $\mathcal{L}\notin \QCMA^{G,\ora}$. 
This completes the proof of \Cref{thm:separation_QMA_QCMA_classical_access}.
\end{proof}
\section{\texorpdfstring{$\QMA$ vs $\QCMA$ under Distributional Oracle}{QMA vs QCMA under Distributional Oracle}}\label{sec:QMA_QCMA}

In this section, we demonstrate a $\QMA$ and $\QCMA$ separation relative to a distributional quantumly-accessible classical oracle.
\begin{lemma}
\label{lem:indist_qma_qcma} 
There exists a family of oracles $\{\ora_n^{b}[H, r]\}_{n\in \mathbb{N},b\in \bit,
H\in \mathcal{H}_n,
r\in\bit^n,}$ 
and a distribution $\mathcal{D}_n$ over 
$\mathcal{H}_n\times \bit^n$ such that the following hold: 
    \begin{enumerate}
        \item {\bf (Distinguishability with Quantum Witness.)} There exists QPT algorithm
            $\A$ and $\poly(n)$-qubit quantum witness $\{\ket{z_{H}}\}_{H}$ such that
            \begin{align*}
                \Pr_{\substack{(H,r)\la \mathcal{D}_n \\b\la \bit}}[\A^{\ora_n^{b}[H,r]}(\ket{z_H},r)=b]\geq 1 - \negl(n).
            \end{align*}
        \item \label{item:ind_with_classical_witness} {\bf(Indistinguishability with Classical Witness.)} For any QPT algorithm
            $\B$ and any polynomial $s$, there is a negligible function $\mu$ such that for any 
            $s(n)$-bit classical witness $\{z_{H}\}_{H}$
            \begin{align*}
                \left|\Pr_{\substack{(H,r)\la \mathcal{D}_n \\b\la \bit}}[\B^{\ora_n^{b}[H,r]}(z_{H},r)=b]-\frac{1}{2}\right|\leq \mu(n).
            \end{align*}
    \end{enumerate}
\end{lemma}

\begin{proof}
    Let  $C\subseteq \Sigma^n$ be the code in \Cref{thm:YZLiu_separation}.
Let $\mathcal{H}_n:=\Func([n]\times \Sigma,\bit)$. 
  For a function $H_n:[n]\times \Sigma \rightarrow \bit$ and $r_n \in \bit^n$, let
  $\ora_{n}^{b}[H,r]$ be an oracle that works as follows:
  if $b=1$, $\ora_{n}^{1}[H,r]$ takes $\vv\in \Sigma^n$ as input and outputs $1$ if $f_C^H(\vv)=r_n$ and outputs $0$ otherwise; 
  if $b=0$, $\ora_{n}^{0}[H,r]$ always returns $0$ for all inputs  $\vv\in \Sigma^n$. 
  Note that $\ora_{n}^0[H,r]$ does not depend on $(H,r)$ at all, but we use this notation for convenience.
  
  \paragraph{Distinguishability with quantum witness.} We use the quantum advice of~\Cref{thm:YZLiu_separation} as a witness.
  Since we can generate a $\vv\in C$ such that $H(\vv)=r$ using the algorithm in~\Cref{thm:YZLiu_separation} with probability~$1-\negl(n)$ over random $H,r$, we query $\ora_n^{b}[H,r]$ with the generated $\vv$, and outputs the query result. If $b=1$, the oracle should return 1 with probability $1-\negl(n)$, else it will always return $0$. 
     
  \paragraph{Indistinguishability with classical witness.}
     
  By the one-way to hiding lemma (\Cref{lem:o2h}), the probability that an unbounded-time algorithm $\B$ can distinguish $\ora_{n}^1[H, r]$ from $\ora_{n}^0[H, r]$ is related to the probability that measuring a random query of $\B$ collapses to an input on which the two oracles differ. The inputs $\vv$ for which $\ora_{n}^1[H, r]$ and $\ora_{n}^0[H, r]$ differ are precisely the ones that satisfy $f_C^H(\vv) = r$. Thus if we define $\M$ to be the adversary that outputs the measurement of a random query of $\B$ we have
  \begin{align*}
      \left|\Pr_{H,r}[\B^{\ora_{n}^{1}[H,r]}(z_{H},r)=1]-\Pr_{H,r}[\B^{\ora_{n}^{0}[H,r]}(z_{H},r)=1]\right|\\
      \leq 2q\sqrt{\Pr\left[f_C^H(\vv) = r \mid \vv \leftarrow \M^{\ora_{n}^{0}[H, r]}(z_{H}, r)\right]}
  \end{align*}
  where $q$ is the number of $\B$'s queries.
  
  Since $\ora_{n}^0[H, r]$ is an oracle with all-zeros, $\M$ can simulate it without access to $H$. Thus from \Cref{item:YZ_easiness_with_quantum_advice} of \Cref{thm:YZLiu_separation}, the RHS is $\negl(n)$. This implies that for any unbounded-time algorithm $\B$ that makes $\poly(n)$ quantum queries and a family of $\poly(n)$-bit classical witness $\{z_{H}\}_{H}$,  
  \begin{align*}
      \left|\Pr_{H,r}[\B^{\ora_{n}^1[H,r]}(z_{H},r)=1]-\Pr_{H,r}[\B^{\ora_{n}^0[H,r]}(z_{H},r)=1]\right|\leq \negl(n).
  \end{align*}
Equivalently, we have
  \begin{align*}
      \left|\Pr_{H,r,b\la \bit}[\B^{\ora_n^{b}[H,r]}(z_{H},r)=b]-\frac{1}{2}\right|\leq \negl(n).
  \end{align*}
\end{proof}



We now use the standard diagonalization argument to translate the indistinguishability result of~\Cref{lem:indist_qma_qcma} to a $\QMA$ and $\QCMA$ separation with respect to a distributional quantumly-accessible classical oracle $\ora$.

\begin{theorem}\label{thm:separation_QMA_QCMA_dist}
    There is a distributional quantumly-accessible classical oracle $\ora$ relative to which $\QMA \neq \QCMA$. 
    

    
    \end{theorem} 
\begin{proof}
Let $\calL$ be a unary language chosen uniformly at random, that is for each $n$, we choose $b_n\la \bit$ independently and put $1^n$ into $\calL$ if and only if $b_n=1$. 
The oracle $\ora=\{\ora_n\}_{n\in \mathbb{N}}$ is chosen as follows where we abuse the notation to write $r_n$ to mean an oracle that takes a null string as input and outputs $r_n$:
\begin{itemize}
    \item If $1^n \in \calL$ (i.e., $b_n=1$), then $\ora_n := (\ora^1_n[H_n, r_n],r_n)$ where $(H_n, r_n)\la \mathcal{D}_n$.
    \item If $1^n \not\in \calL$ (i.e., $b_n=0$), then $\ora_n := (\ora^0_n[H_n, r_n],r_n)$ where $(H_n, r_n)\la \mathcal{D}_n$.
\end{itemize}

We start with proving that $\calL \in \QMA^\ora$  with probability 1 over the choice of $\{(b_n,H_n,r_n)\}_n$.  The verifier $V$ works as follows: It first queries $\ora$ at point $0^n$, obtains the random string $r_n$, then calls the algorithm in~\Cref{lem:indist_qma_qcma}.  Since for each $n$, if $b_n=1$, the algorithm $\A^{\ora_n^{1}[H_n,r_n]}(\ket{z_{H_n}},r_n)$ will return 1 with probability $1-\negl(n)$, and if $b_n=0$, the $\A^{\ora_n^{0}[H_n,r_n]}(\ket{z},r_n)$ will always return 0 for all witness $\ket{z}$. Thus for each $n$ the verifier will fail with probability $\negl(n)$. By applying the Borel-Cantelli lemma and following the same arguments as in~\Cref{thm:separation_poly_qpoly}, we can prove that $\calL\in\QMA^\ora$. 
Note that our final oracle distribution $\F$ will fix $H_n$ for each $n$, thus here we allow the witness $\ket{z_{H_n}}$ to depend on $H_n$.

Next, we prove $\calL \not\in \QCMA^\ora$ with probability $1$ over the choice of $\{(b_n,H_n,r_n)\}_n$. Fix a QPT machine $V$ that takes $\ell(n)$-bit classical witness for some polynomial $\ell$  and let $S_V(n)$ be the event that $V^\ora$ succeeds on $1^n$, that is either
    \begin{itemize}
        \item $1^n \in \calL$ and there exists classical witness $w_\ora\in \bit^{\ell(n)}$ 
        such that
            $V^\ora$ accepts $(1^n,w_\ora)$ 
            with probability at least $\frac{2}{3}$, or
        \item $1^n \not \in \calL$ and $V^\ora$ accepts $(1^n,w)$ with
            probability at most $\frac{1}{3}$ for all $w\in \bit^{\ell(n)}$.
    \end{itemize}
    To be precise, 
    \begin{align*}
        S_V(n)=[&\exists w\in \bit^{\ell(n)} : \Pr[V^{\ora^{b_n=1}}(1^n,w)=1]\geq 2/3] \\
        &\vee[\forall w\in \bit^{\ell(n)} : \Pr[V^{\ora^{b_n=0}}(1^n,w)=1]\leq 1/3]
    \end{align*}
    We first show that for any QPT algorithm $V^\ora$, there is a distinguishing algorithm $\B^{\ora_n^{b_n}[H_n,r_n]}$ in~\Cref{lem:indist_qma_qcma} that has the same accept probability as $V^\ora$ for all given $(H_1,H_2,\dots)$.  $\B^{\ora_n^{b_n}[H_n,r_n]}$ takes the witness $z_\ora$, and it hardcodes all other $H_i,b_i(i\neq n)$ in its program, and it simulates the behaviour of $V^\ora$ by randomly choosing $r_i$ for $i\neq n$, and calculates $O_i^{b_i}[H_i,r_i]$ for $i\neq n$ oracle queries, and when querying $n$, it queries its own oracle. In the end, $\B$ sets its output to be the same as $V$. It can be seen that
    \begin{align*}
        \Pr_{H_n,r_n}\left[\B^{\ora_n^{b_n}[H_n,r_n]}(w_\ora,r_n)=1\right]=\Pr_{\{r_i\}_i,H_n}\left[V^\ora(1^n,w_\ora)=1\mid \{H_i,b_i\}_{i\neq n}\right].
    \end{align*}
    Notice that on the left-hand side we can view $w_\ora$ as a witness $z_{H_n}$ as in~\Cref{lem:indist_qma_qcma}.
    \item Now we prove that 
    there exists universal constant $c < 1$ such that for any QPT verification algorithm $V$, there exist infinitely many $n$ such that
    $$\Pr_{\{(b_n,H_n,r_n)\}_n}[S_V(n)] < c.$$
    
    By our observation above, if we can prove that for all unbounded quantum algorithms $\B$, the following inequality holds:
    \begin{align*}
        \Pr_{H_n,r_n,b_n}[[E_1\wedge b_n=1]\vee [E_2\wedge b_n=0]]<c,
    \end{align*}
    where
    \begin{align*}
         &E_1=\exists w\in \bit^{\ell(n)} : \Pr[\B^{\ora_n^{1}[H_n,r_n]}(w,r_n)=1]\geq 2/3,\\
        &E_2=\forall w\in \bit^{\ell(n)} : \Pr[\B^{\ora_n^{0}[H_n,r_n]}(w,r_n)=1]\leq 1/3,
    \end{align*}
    then $\Pr[S_V(n)]$ is also bounded by an averaging argument.
    We can see that
    \begin{align*}
        \Pr[[E_1\wedge b_n=1]\vee[E_2\wedge b_n=0]]=\frac{1}{2}(\Pr[E_1]+\Pr[E_2]).
    \end{align*}
    It can be seen that
    \begin{align*}
        \Pr[E_1]=\Pr[E_1\wedge\neg E_2]+\Pr[E_1\wedge E_2]. 
    \end{align*}
     We prove $\Pr[E_1\wedge E_2] \le \frac{4}{5}(1+2\mu(n))$ as follows. 
     By \Cref{item:ind_with_classical_witness} of \Cref{lem:indist_qma_qcma}, for all polynomial sized $w_{H_n}$,   
     \begin{align*}
     \Pr_{H_n,r_n}[\B^{\ora_{n}^1[H_n,r_n]}(w_{H_n},r_n)=1]-\Pr_{H_n,r_n}[\B^{\ora_{n}^0[H_n,r_n]}(w_{H_n},r_n)=1]\leq 2\mu(n).
  \end{align*}
  By a standard averaging argument,\footnote{Concretely, add $1$ to both sides, apply Markov's inequality, and then subtract $1$ from both sides.
  Then we see that for at most $\frac{4}{5}(1+2\mu(n))$ fraction of $(H_n,r_n)$, we have $\Pr[\B^{\ora_{n}^1[H_n,r_n]}(w_{H_n},r_n)=1]-\Pr[\B^{\ora_{n}^0[H_n,r_n]}(w_{H_n},r_n)=1]\ge 1/4$. 
  }
  we can conclude that for at least 
  $1-\frac{4}{5}(1+2\mu(n))$  
  fraction of $(H_n,r_n)$, 
  $$\Pr[\B^{\ora_{n}^1[H_n,r_n]}(w_{H_n},r_n)=1]-\Pr[\B^{\ora_{n}^0[H_n,r_n]}(w_{H_n},r_n)=1]<1/4,$$
  implying that if $E_1$ occurs, $E_2$ does not occur, i.e., there exists $w_{H_n}\in \bit^{\ell(n)}$ such that
  \begin{align*}
      \Pr[\B^{\ora_{n}^0[H_n,r_n]}(w_{H_n},r_n)=1]>\frac{2}{3}-\frac{1}{4}>\frac{1}{3}.
  \end{align*}
     Thus we have that 
     $\Pr[E_1\wedge E_2] \le \frac{4}{5}(1+2\mu(n))$, 
     giving us
    \begin{align*}
     \frac{1}{2}(\Pr[E_1]+&\Pr[E_2])\leq  \frac{1}{2}(\Pr[E_1\wedge\neg E_2]+\frac{4}{5}(1+2\mu(n))+\Pr[E_2])\\
        &\leq \frac{1}{2}(\Pr[\neg E_2]+\frac{5}{6}+\Pr[E_2])=\frac{11}{12}.
    \end{align*}
    Thus, we can set $c:=11/12$. 

We now show that 
$$\Pr_{\{H_n,r_n,b_n\}_n}[\exists V\;S_V(1) \land S_V(2) \land \dots] = 0.$$
Consider a sequence of input lengths $n_1, n_2, \dots$, such
that the indistinguishability condition of~\Cref{lem:indist_qma_qcma} holds for $n_i$ and $n_i \geq T(n_{i-1}) + 1$, where $T(n)$ is the running time of $V$ on input of length $n$. This means that when $V$'s input length is $n_{i-1}$, it cannot touch the oracle on input length $\ge n_i$. This guarantees that $\Pr[S_V(n_i) \mid S_V(n_{j})] = \Pr[S_V(n_i)]$ for all $i > j$.

We can now show that the probability that $V$ succeeds on all inputs is equal to $0$ over the choices of $H_n, r_n, b_n$.
\begin{align*}
    &\Pr[S_V(1) \land S_V(2) \land \dots] \\
    &\leq \Pr\left[ \bigwedge_{i} S_V(n_i)\right] \\
    &= \Pr[S_V(n_1)] \cdot \Pr[S_V(n_2) \mid S_V(n_1)] \cdot \dots \\
    &\leq c \cdot c \cdot \dots \\
    &= 0
\end{align*}
Since there are countably many $\QCMA$ machines, we have that
$$\Pr_{\{H_n,r_n,b_n\}_n}[\exists V\;S_V(1) \land S_V(2) \land \dots] = 0.$$
Thus $\calL \not\in \QCMA^\ora$ with probability $1$ over the choice of $\{H_n,r_n,b_n\}_n$.

We conclude that by fixing $H_n,b_n$ for each $n$, we can obtain a language $\calL$ in $\QMA^\ora$ but not in $\QCMA^\ora$, where $\ora$ is now a distribution over $\{r_n\}_n$.

\end{proof}

\section{One-Way Communication Complexity}\label{sec:communication} 
We observe that the results of \cite{FOCS:YamZha22,EC:Liu23} (\Cref{thm:YZLiu_separation}) can be seen as a separation of classical and quantum one-way communication complexity.
Consider the following protocol between Alice and Bob where we use the notations defined in \Cref{sec:YZ}.
\begin{description}
\item[Alice's input:] A truth table of a function $H:[n]\times \Sigma\ra \bit$.
\item[Bob's input:] A string $r=(r_1,...,r_n)\in \bit^n$.
\item[One-way communication:] Alice sends a classical or quantum message $m$ to Bob.
\item[Bob's output:] Bob outputs $\vv=(\vv_1,...,\vv_n)\in C$. We say that Bob wins if $H_i(\vv_i)=r_i$ for all $i\in [n]$. 
\end{description}

Quantum (resp. classical) one-way communication complexity is defined to be the minimum length of the quantum (resp. classical) message $m$ that enables Bob to win with high probability (say, $2/3$).  
\Cref{thm:YZLiu_separation} directly means that quantum one-way communication complexity is $\poly(n)$ but classical one-way communication complexity is super-polynomial in $n$ (one can see that it is actually subexponential in $n$ from its proof).\footnote{Strictly speaking, \Cref{item:YZ_easiness_with_quantum_advice} of \Cref{thm:YZLiu_separation} only ensures the existence of a quantum communication protocol with $\poly(n)$-qubit communication that works on average over random $H$.
In the standard definition of one-way communication complexity, we require a protocol to work for all inputs. For ensuring that, we have to rely on the perfectly correct version of \cite{FOCS:YamZha22} that can be found at the end of \cite[Section 4]{FOCS:YamZha22}.
} 
This gives a new super-polynomial separation between classical and quantum one-way communication complexity.
Such a separation is already known back in 2004 by Bar-Yossef, Jayram, and Kerenidis~\cite{STOC:BarJayKer04} where they showed it based on a problem called the hidden matching problem. 
However, our protocol has the following two interesting features compared to theirs. 

First, Bob's input length is exponentially smaller than Alice's input length. 
Why is that interesting? To give the context, we review the following theorem shown by Aaronson~\cite{Aaronson07}.
\begin{theorem}
    For any (possibly partial) boolean function $f:\bit^N\times \bit^M \rightarrow \bit$,  
    $\mathcal{R}^{1}(f)=O(M \mathcal{Q}^{1}(f))$.
\end{theorem}
Here, $\mathcal{Q}^{1}$ and $\mathcal{R}^{1}$ mean the quantum and classical randomized bounded-error one-way communication complexity, respectively. 
This theorem means that we cannot have a large quantum-classical separation for boolean functions if Bob's input length is small. 
We circumvent this barrier by considering \emph{relations} rather than functions.\footnote{Interestingly, a similar observation was made by a concurrent work~\cite{ABK23} where they show a variant of the hidden matching problem with short Bob's input length.} 
This is reminiscent of \cite{FOCS:YamZha22} that overcomes the barrier of Aaronson-Ambainis conjecture~\cite{AA14} by considering search problems rather than decision problems. 

Second, by \Cref{thm:YZLiu_separation_classical_query}, one can see that the hardness with classical communication remains to hold even if we allow Bob to classically query Alice's input. On the other hand, the hidden matching problem becomes completely easy if we allow such classical queries. This property is the key to show $\BQP/\poly\neq \BQP/\qpoly$ and $\QMA\neq \QCMA$ relative to classically-accessible classical oracles.

\section{Acknowledgements}

The authors would like to thank Shalev Ben-David and Srijita Kundu for pointing out an issue with a proof in a prior version of this paper.

\bibliographystyle{alpha} 
\bibliography{abbrev3,crypto,reference}

\end{document}